%% file: FaultTolerentComputing.tex
\documentclass[11pt]{article}
\usepackage{microtype}
\usepackage{xspace}
\usepackage{fullpage}
\usepackage{amsmath}
\usepackage{amsthm}
\usepackage{amssymb}
\usepackage{times}
\usepackage{url}
\usepackage{paralist}
\usepackage[usenames]{color}
\usepackage[table,svgnames]{xcolor}
\usepackage[colorlinks=true,linkcolor=black,citecolor=MidnightBlue,urlcolor=MidnightBlue,hypertexnames=false]{hyperref}
\usepackage{mathtools}
\usepackage{comment}
\usepackage{thmtools}
\usepackage{thm-restate}

\usepackage{enumerate}
\usepackage{algorithm}
\usepackage{algorithmicx}
\usepackage{algpseudocode}
\usepackage{cleveref}
\crefrangeformat{line}{lines~#3#1#4--#5#2#6}
\usepackage{refcount}
\usepackage{tabularx}
\usepackage{cite}

\usepackage{tikz}
\usetikzlibrary {shapes,patterns,patterns.meta}
\usetikzlibrary{arrows.meta, positioning}

\usepackage{sidecap}
\sidecaptionvpos{figure}{t}

\title{Two for One, One for All:\! Deterministic LDC--based Robust Computation in Congested Clique}

\input{macros}
\begin{document}

\date{}
\author{
Keren Censor-Hillel \\
	\small Technion \\
	\small ckeren@cs.technion.ac.il \\
	\and
	Orr Fischer \\
	\small Bar-Ilan University \\
	\small orr.fischer@biu.ac.il \\
	\and
	Ran Gelles \\
	\small Bar-Ilan University \\
	\small ran.gelles@biu.ac.il \\
	\and
	Pedro Soto \\
	\small Virginia Tech \\
	\small pedrosoto@vt.edu \\}
\maketitle

\begin{abstract}
We design a deterministic compiler that makes any computation in the \textsf{Congested Clique} model robust to a constant fraction $\alpha<1$ of adversarial crash faults. 
In particular, we show how a network of~$n$ nodes can compute
any circuit of depth~$d$, width~$\omega$, and gate total fan~$\Delta$,  in
$d \cdot \lceil \frac{\omega}{n^2}+\frac{\Delta}{n} \rceil \cdot 2^{O(\sqrt{\log{n}}\log\log{n})}$
rounds in such a faulty model.
As a corollary, any $T$-round \textsf{Congested Clique} algorithm 
can be compiled into an algorithm that completes in $T^2 n^{o(1)}$ rounds in this model.

Our compiler obtains resilience to node crashes by coding information across the network, and its main underlying observation is that we can leverage locally-decodable codes (LDCs) to maintain a low complexity overhead, as these allow recovering the information needed at each computational step by querying only small parts of the codeword, instead of retrieving the entire coded message, which is inherent when using block codes.

The main technical contribution is that because erasures occur in known locations, which correspond to crashed nodes, we can \emph{derandomize} classical LDC constructions by deterministically selecting query sets that avoid sufficiently many erasures. 
Moreover, when decoding multiple codewords in parallel, our derandomization \emph{load-balances} the queries per-node, thereby preventing congestion and maintaining a low round complexity.

Deterministic decoding of LDCs presents a new challenge: the adversary can target precisely the (few) nodes that are queried for decoding a certain codeword.
We overcome this issue via an adaptive doubling strategy: 
if a decoding attempt for a codeword fails, the node doubles the number of its decoding attempts.
We employ a similar doubling technique when the adversary crashes the decoding node itself, replacing it dynamically with two other non-crashed nodes. 
By carefully combining these two doubling processes, we overcome the challenges posed by the combination of a deterministic LDC with a worst case pattern of crashes.
\end{abstract}


\widowpenalty=5000
\clubpenalty=5000
\linepenalty=500
\sloppy

\section{Introduction}
\label{sec:intro}
Robustness is a crucial component in the design of distributed algorithms, as faults lie at the heart of distributed computing environments. Thus, addressing various types of failures has been heavily studied, see, e.g., seminal results covered in classic books on distributed computing \cite{AttiyaWelch, Lynch96, peleg2000distributed}. 
In this paper, we focus on \emph{node crashes} in the \cliquefull model (introduced by Lotker et al.~\cite{LPPP05}), in which the computing devices use bandwidth-restricted point-to-point communication over a complete network graph. 
So far, despite abundant research in this model (see related work in \Cref{subsec:related}), relatively less attention was given to fault tolerance questions in this useful model~\cite{AMPV22,KMS22,KM23,MR23,FP25}.

In some settings, such as the \congest model or the work in \cliquefull of \cite{KMS22,KM23}, if node crashes are allowed then the requirement is that the output of the computation corresponds only to the inputs of non-crashed nodes. In contrast, Censor-Hillel and Soto~\cite{CS25} show how to avoid losing any information in the \cliquefull despite node crashes. They do this by having each node encode its input and the results of any subsequent local computations using erasure correction codes,  split the codewords to pieces and distribute them to the nodes of the network. 
Upon a node crash, the other nodes collect the pieces of the relevant codeword and decode it in order to continue the computation. 
This approach implies that any task
over a network of $n$ nodes (with the typical setup of $O(n\log{n})$ bits of input per node which we will also work with here)
can be computed in~$O(n)$ rounds despite crashes:
after each node distributes its inputs in an encoded way, all other nodes gather all the pieces and reconstruct the entire $O(n^2\log n)$-bit input, from which they can locally compute the output.
The work~\cite{CS25} beats this bound for certain tasks whose circuit representation has good properties. 
While their work can achieve fast resilient algorithms for well-behaved circuits, e.g., retain the $\tilde{O}(n^{1/3})$-round complexity for matrix multiplication~\cite{CKPLPS15} even with faults,
for a \emph{general circuit}, 
their resilient algorithm may incur a multiplicative overhead of $n$~rounds, which is worse than the solution that simply learns the entire encoded input.

In this paper, we provide a compiler that makes \cliquefull algorithms robust to a constant fraction of node crashes, by computing general circuits faster. This results in a compiler with a complexity of $T^2 n^{o(1)}$ rounds for a \cliquefull algorithm of $T$ rounds, and thus it  beats the ${O}(n)$-round solution when $T=o(n^{1/2-o(1)})$. In the crash model we consider, node crashes occur at the start of a round, and the adversary may crash up to $\alpha n$ nodes in total throughout the algorithm, for a fault parameter $\alpha \in [0,1)$.

\begin{theorem}
\label{cor:main-formal}
Let $\mathsf{ALG}$ be any \cliquefull algorithm that completes in $T$~rounds. 
Then, for any $\alpha\in[0,1)$, there is an equivalent algorithm $\mathsf{ALG}'$ that is resilient to an $\alpha$ fraction of crashes and
completes in $T^2n^{o(1)}$~rounds.
\end{theorem}

A concrete example for an application of \Cref{cor:main-formal} is computing exact single-source-shortest-paths (SSSP) in weighted undirected graphs. The $\tilde{O}(n^{1/6})$-round algorithm of \cite{CDKL21} translates by our compiler to an algorithm that completes in  $n^{1/3+o(1)}$ rounds, even if any $\alpha n$ nodes may crash during its execution, for any $\alpha<1$.

At the heart of our compiler is a faster deterministic algorithm for computing general circuits in the faulty model. The following is our main technical contribution.

\begin{restatable}{theorem}{MainWireThm}
\label{thm:circuit}
Let $C$ be circuit of depth $d$, 
max total-fan $\Delta$, and width $\maxWidth$. 
Then, for any $\alpha\in[0,1)$
there exists a deterministic \cliquefull  algorithm for computing the output of~$C$ in the presence of~$\alpha n$ crashes, 
whose round complexity is $d \cdot \lceil \frac{\maxWidth}{n^2}+\frac{\Delta}{n} \rceil \cdot 2^{O(\sqrt{\log{n}}\log\log{n})}$.
\end{restatable}

While our transformation also goes through computing circuits, our algorithm differs from that of \cite{CS25} in two main aspects, which we discuss next.

The first regards the initialization phase of the resilient algorithm.
Note that if the adversary fails even a single node before the start of the computation, then this node's input is lost forever.
The solution taken by~\cite{CS25} is  promising $1/(1-\alpha)$ \emph{quiet} rounds 
where no crashes can happen.
These quiet rounds can be used 
to encode and distribute the nodes inputs.
We take a different approach: we consider the inputs at the start of the computation 
as \emph{encoded versions} of the inputs to the \cliquefull model (similar to Spielman's coded computation~\cite{Spielman96}). 
Since our robust algorithm is such that the outputs are also encoded in this manner, we get \emph{composability} for free---we can simply start another computation after completing a former one.

Second, instead of using block error correction codes, we use \emph{locally decodable codes (LDCs)}~\cite{KT00,Yekhanin12}, which allow a node to query only a small number of other nodes in order to decode only the information it needs for its next local computation. This yields a dramatic improvement in the round complexity of computing a circuit because of its huge effect on congestion. 
Working with LDCs gives rise to new challenges, because decoding an LDC codeword can be easily targeted by the adversary which can crash so many nodes. The crux of the proof of \Cref{thm:circuit} shows how to overcome this, and that in fact it can be implemented in a \emph{deterministic} manner. %

\medskip

\Cref{cor:main-formal} is immediately established by combining our robust computation of circuits in  \Cref{thm:circuit} with the natural conversion of any \cliquefull algorithm to a circuit (see a formal proof in, e.g., \cite{CS25}).

\begin{restatable}{lemma}{CliqueToCirc}
\label{lem:cliq_to_circ}
Let $\mathsf{ALG}$ be a \cliquefull algorithm that computes a function $f$ in $T$ rounds.
Then, there is a  
circuit
with depth $d=2T+1$, width $\maxWidth = \Theta(Tn^2\log{n})$, and maximal gate total fan of $\Delta=O(Tn\log{n})$, whose inputs are all the $O(n^2\log n)$ input bits of the nodes and whose outputs are the outputs of the nodes  after running $\mathsf{ALG}$.
\end{restatable}

\subsection{Technical Overview}
\label{sec:overview}
Our main theorem states that we can robustly compute a circuit~$C$ of depth~$d$ and arbitrary gates, despite a possible (worst-case) crash of an $\alpha$ fraction of the nodes.
Towards this end, the network computes the gates of~$C$, layer by layer, where each node is assigned some of the gates in each layer in a dynamic manner that adapts to node failures. 

As mentioned above, when a node crashes, it can no longer send messages, and thus any information that was privately held in that node's memory is lost forever. 
To be resilient to crashes and avoid losing important information, we need to store information ``in the network'' so that it can be retrieved despite a constant fraction of node crashes. Indeed, \cite{CS25} used block error correction codes (ECCs) to encode data and split the resulting codewords across the network. This way, each piece of information can be retrieved by querying all the nodes for their part of the codeword; the decoding succeeds even if 
a constant fraction of the nodes crash, where the constant depends on the strength of the ECC in use. The drawback is that even if a single bit of information is needed from a coded string, its entire codeword needs to be decoded.

Our starting point is that we replace block codes with Locally Decodable Codes (LDCs).
Informally, such codes allow decoding specific \emph{parts} of the message, rather than decoding the entire message. Furthermore, decoding does not require obtaining the entire (corrupted) codeword, but rather queries relatively small parts of it (a subpolynomial number of symbols), while still guaranteeing  correct decoding of the desired symbol with a high probability.

Hence, switching to LDCs benefits our algorithm in the sense that nodes make \emph{less queries} in order to retrieve \emph{exactly} the information they need for the computation. 
The advantage of LDCs over standard (i.e., block-codes) ECCs becomes clear when considering the case where some node is assigned to compute a gate with $n$~inputs,  each of which is stored in a different codeword. 
With standard ECCs, this means that the node must access all $n$ codewords and retrieve all information bits, i.e., $n^2$ bits, assuming $n$-bit codewords, which causes high congestion.

\subparagraph*{Algorithm Overview.}
The high-level idea of our robust circuit computation algorithm is as follows. 
Consider a circuit~$C$, whose inputs are distributed over the network in an encoded manner using some LDC code.
The network  computes~$C$ layer by layer. 
That is, let $\gates(1)$ be the first layer of gates in~$C$, i.e., all the gates whose inputs are the inputs of~$C$. 
We first distribute the tasks of computing these gates across the (non-crashed) nodes so that each node is assigned roughly the same number of gates to compute. Each node then attempts to compute all the gates allocated to it. 
To do this, the node first obtains the inputs for each gate assigned to it by decoding the corresponding inputs of~$C$, which are stored in the network using codewords of an LDC.
If this information retrieval is successful, the node computes the outputs of the gates allocated to it, and then ``stores'' them in the network by encoding them with an LDC and distributing the codewords to the nodes in the network. 
Once the first layer is computed and stored in the network, the nodes continue to compute the second layer of gates in~$C$, denoted $\gates(2)$, which includes all gates whose inputs are either the inputs of $C$ or the outputs of the gates in the first layer, $\gates(1)$. This continues until all the outputs of~$C$ are computed and stored in the network. 
Naturally, crashes that occur during the algorithm may prevent the network from completing the computation of a specific layer and  progressing to the next layer, which we discuss next.

\subparagraph*{Overcoming Crashes I.} 
In our robust circuit computation algorithm, each gate is assigned to a dedicated node responsible for its computation. 
If that node crashes, the gates assigned to it remain uncomputed. 
A trivial solution is to reassign any uncomputed gate in the current layer of~$C$ (whose original node is crashed) to a new node that is not crashed. 
However, the adversary could then crash this new node and eventually cause a delay of $\alpha n$ rounds, which is extremely expensive.
Our strategy is different: when a node crashes, we reassign its gate(s) to \emph{two} fresh nodes.
If both of those nodes crash, we again double the redundancy, forcing the adversary to double its effort to keep the gate uncomputed. 
After at most $\log n$ such iterations, every gate is guaranteed to be computed by at least one live node. 
This progressive doubling remains feasible without causing excessive congestion because of two factors: 
First, the nodes that do not crash successfully compute and store their assigned gates. These nodes are now available to take over the gates of the crashed nodes.
Second,  the adversary cannot (effectively) corrupt too many nodes in the same round:
If the adversary crashes too many nodes during a short period of time, we call this step \emph{overwhelmingly faulty}, and simply restart the computation of this layer with the remaining nodes. While this translates to no progress, it reduces the adversary's budget of crashes and hence cannot occur too many times.

\subparagraph*{Deterministic LDCs and Congestion.}
The decoding algorithm of LDCs is \emph{inherently random}. 
Indeed, if a fixed (small) number of codeword symbols are queried during a decoding attempt and these symbols are corrupted, then decoding is certainly impossible.
If so, the code cannot correct a constant fraction of corruptions, as would be normally expected. 
In particular, given a budget of  $\alpha n$ node crashes, an all-knowledgeable adversary may be able to crash a subset of nodes in a way that prevents any meaningful progress of the deterministic  computation.

Despite the above conundrum, our robust algorithm is fully deterministic.
In particular, we derandomize the LDC decodings performed throughout the computation while maintaining resilience to~$\alpha n$ node crashes.
Our derandomization relies on two important properties, specific to our model. First, when a node crashes, all other nodes are aware of this event because the crashed node does not send any messages  from the round in which it crashed. 
This allows the remaining nodes to maintain a  consistent view of the crashed and alive set of nodes.
Second, crashed nodes that are queried for their respective parts of the LDC codeword do not reply, and thus the LDC decoding algorithm is missing some parts of the queried codeword, known as \emph{erasure} corruptions. These are easier to correct than when the codeword contains incorrect information. 
The combination of erasure corruptions and knowledge of which nodes have crashed in each round allows a decoder to predict  whether a specific set of queries will result in successful decoding.
Thus the decoder can pick a set of queries that is guaranteed to succeed if no further nodes crash.

While the above idea derandomizes the (inherently random) LDC decoding algorithm, it creates a new challenge regarding the resulting congestion.
To illustrate this challenge, suppose that a node performs the above deterministic selection of queries \emph{separately} for each piece of information it wants to retrieve. 
Then, it may end up querying the same subset of nodes over and over again, thus causing a large congestion. 
Randomized decoding averts this issue by querying a set of nodes in a near-uniform distribution. However, even if we could deterministically replicate this querying distribution, we face again the issue mentioned above, where many of these queries are erased and do not lead to a correct decoding.

Nevertheless, our  analysis, which is based on the probabilistic method, shows that
it is possible to select query sets for multiple (independent) LDC-decoding instances in the presence of a constant fraction of erasures in positions known to the decoding algorithm, 
so that the following hold simultaneously:
(i) each decoding instance successfully decodes the correct information, and 
(ii) the congestion per queried node is small, i.e., the queries are well-distributed over the network. 
To show the latter, we analyze an equivalent bins-into-balls experiment, showing that the event  that too many balls (queries) aggregate in one specific bin (node) happens with small probability. Union-bounding over all nodes  keeps the probability of the bad event below~1, thus proving the existence of good query sets that avoid congestion.

\subparagraph*{Overcoming Crashes II.}
The above discussion implies that the adversary \emph{cannot} select a set of $\alpha n$ nodes to crash for invalidating many of the LDC decoding attempts throughout the computation, as long as these indices are known to the nodes. 
However, the adversary may decide to crash nodes \emph{after} a decoder fixes its selection of nodes to query in a given round, as this selection depends only on nodes that are crashed \emph{prior} to that round. 
To overcome this problem, the nodes dynamically increase the number of times they attempt to LDC decode each piece of information, according to corruptions made so far.
Namely, if the decoding of some LDC-encoded information fails due to \emph{new} corruptions of the queried nodes,
then the decoding node performs \emph{two} independent decoding attempts. 
These new queries depend on all the crashes so far,  and in particular, on the ``new'' crashes that invalidated the original decoding attempt.
If these two attempts  fail as well (due to new crashes that occur after the nodes queried by these two attempts are decided), the decoding node doubles its number of attempts again, and so on.
Overall, after a logarithmic number of doublings, this approach potentially causes a large, near-linear number of LDC decoding attempts, and  the adversary can only fail a constant fraction of them without exceeding its budget. 
Note that it only takes one successful attempt to move on, so the adversary must fail all attempts of a single codeword to prevent progress.

\subparagraph*{Recap.}
We can now summarize the overview of our robust circuit computation. 
For a given circuit~$C$, the computation goes layer by layer, where computing a layer of~$C$ means: (1) assigning uncomputed gates of the layer to non-crashed nodes; (2) retrieving the inputs to the gates of this layer, that are stored in the network via an LDC during the computation of previous layers; (3) computing the gates; (4) storing the outputs of the gates via an LDC. 
Technically speaking, this computation of each layer is done in two nested loops:
The external loop doubles, in each iteration, the number of nodes  responsible for computing some uncomputed gate (we call this loop \emph{the $\outStep$ loop}).
The internal loop doubles, in each iteration, the number of independent decoding attempts each node makes for each uncomputed gate assigned to it (we call this loop \emph{the $\inStep$ loop}). 
\Cref{sec:Simulation} fully details our algorithm, \Cref{sec:analysis} analyzes its correctness and complexity, and \Cref{sec:derandom} shows our derandomization of LDCs in our setup.

\subsection{Additional Related Work}
\label{subsec:related}
Since its introduction for faster MST computation \cite{LPPP05}, the \cliquefull model has been extensively explored during recent decades for various tasks. The MST complexity was eventually shown to be constant \cite{Nowicki21a} following a beautiful line of work \cite{HegemanPPSS15, GhaffariP16, Korhonen16, Jurdzinski018}. Additional examples include routing \cite{Lenzen13}, coloring \cite{Parter18, ParterS18, CFGUZ19, BKM20, CzumajDP21, CoyCDM23}, subgraph finding \cite{DLP12, IzumiG17, Pandurangan0S18, FischerGKO18, Censor-HillelFG20, Censor-HillelGL20, Censor-HillelFG22}, and many more. Hardness of obtaining lower bounds in this model is established in \cite{DKO13}.

Fault-tolerance in the \cliquefull model was explored by 
\cite{KMS22,KM23} for graph realization problems under crash-faults,
by \cite{AMPV22,MR23} for recognizing connectivity and 
hereditary properties under Byzantine faults, and by~\cite{FP23,FP25} for general computations under edge faults. In particular, \cite{FP25} employs LDCs as means to concentrate information from many nodes into few. 

Coding theory, in various forms, has been extensively used in many other areas of distributed computing, including: distributed zero-knowledge \cite{GPP25, BKO22}, proof labeling schemes \cite{fischer2022explicit, Censor-HillelH24}, the beeping model \cite{Davies23a, AshkenaziGL22,GKKM25arXiv}, and distributed interactive proofs \cite{NPY20}.

\section{Preliminaries}
\label{sec:prelim}
    For an integer $n\ge1$ we denote $[n]=\{1,2,\ldots,n\}$. All logarithms are taken to base~2 unless otherwise mentioned. 
    We say that an event occurs with high probability (in~$n$, which is usually implicit) if its probability is at least $1-1/n^{10}$.
    For a string $x$ and for any $i\in [|x|]$, let $x[i]$ denote the $i$-th symbol of~$x$.

\subsection{Computation Model}
\label{sec:prelim:model}
Suppose a \cliquefull network, where $n$~nodes, $v_1,\dots, v_n$, communicate in synchronous rounds by exchanging $b\log{n}$-bit messages in an all-to-all fashion, for some constant $b\in\mathbb{N}$. 
Throughout the computation, an adversary may choose to crash up to $\alpha n$ nodes, where the constant $\alpha\in[0,1)$ is a parameter of the model.
A crashed node does not send any messages starting from the round in which it is crashed. 
The corruption is  \emph{worst case}:  
an all-knowledgeable adversary bases its decision on which nodes to fail on all its available information, including the algorithm that the nodes execute, their inputs, and their local randomness (if any).
Note that all nodes know which nodes are non-faulty at the end of each round, denoted as the set~$\alive$ (to indicate that they are \emph{alive}).

The compiler of \Cref{alg:main} is completely deterministic, in the sense that it does not add any new randomness. 
In particular, 
The robust algorithm $\mathsf{ALG'}$ remains  deterministic if $\mathsf{ALG}$ was deterministic, and similarly, it is  randomized if $\mathsf{ALG}$ was so. In the latter case, we assume that the randomness of $ALG$ is given to the representing circuit~$C$ as input.

\subparagraph*{Coded inputs and outputs.}
In the \cliquefull model, it is common to refer to problems in which each node holds a private input $x$ of $O(n\log n)$~bits. 
As explained in the introduction, in our faulty model, the inputs must be coded in a way that prevents them from being lost if some node crashes before the first round of communication.
We thus consider inputs as encoded via an LDC code (see \Cref{def:LDC} in \Cref{sec:prelim:LDC} below), and the respective codeword is distributed across the nodes of the network.
We employ a $(q,\delta,\epsilon)$-LDC code with block length~$n$, whose parameters will be specified later. 
Such a code guarantees that every bit of the input~$x$ (of each node) can be retrieved by querying $q$ nodes with probability $1-\epsilon$ over a uniform choice of the randomness string, even if $\delta n$ of the nodes have crashed. It will be the case that $\alpha<\delta$. 

We require the output to be stored in the network in a similar way: the outputs should be encoded via an LDC code whose resulting codewords are split among the nodes, so that it is possible to retrieve each bit of the output despite $\delta n$ crashes.  
This choice allows the composition of computations, where the outputs of one computation are the inputs of the next. 

\subsection{Layered Circuits}
\label{sec:circuit}
We identify a circuit~$C$     with a directed acyclic graph $C=(V_C,E_C)$ in which every gate is associated with a node, and every wire connecting gates is associated with an edge.
Each bit-input to the circuit (an input gate) is associated with a leaf node and each output of the circuit (an output gate) with a root node. Other nodes are associated with the computational gates of the circuit. 
We say that a node~$g\in V_C$ \emph{depends} on a node $g'\in V_C$ if there is a directed path from $g'$ to~$g$.
The notion of gate dependencies induces \emph{layers} in the circuit, where all input gates are in layer~$0$, and a gate~$g$ is placed in layer~$i$ if $i$ is the minimal integer such that $g$ depends only on nodes in layers at most~$i-1$.
For a gate~$g\in V_C$, we denote by $\layer(g)$ the layer of~$g$, and let $\gates(i)=\{g\in V_C \mid \layer(g)=i\}$ 
be the set of all the gates in layer~$i$. 
The depth of~$C$, denoted $d=d(C)$, is defined to be $\max_{g\in V_C} \layer(g)$.
We denote by $\wires(i)=\{(u,v)\in E_C \mid \layer(u)=i \}$ 
the set of all wires that go out of the gates in layer~$i$.
Note that $\wires(0)$ are the inputs to the circuit. 
For a gate~$g$, denote by $\fanin(g)$ its in-degree and by $\fanout(g)$ its out-degree; let $\fantotal(g)=\fanin(g)+\fanout(g)$ be its \emph{total fan}.
The \emph{width} of the circuit, $\maxWidth = \maxWidth(C)$, is defined as the maximal number of outgoing wires of any layer,
$\maxWidth = \max_i (\wires(i))$. We assume throughout that all parameters of the circuit are polynomial in the size of the network, i.e. $d,\maxWidth,\Delta = O(\poly(n))$. This fits the case where $C$ represents a \cliquefull algorithm with $O(n\log n)$ bits of input per node. 
We note, however, that the statement of  \cref{thm:circuit} holds, up to logarithmic terms, for any parameters~$d$, $\maxWidth$, and~$\Delta$.

    \begin{SCfigure}[0.6]
        \input{circuit.tex}
        \caption{An example of  a circuit $C$ of depth~$d=3$ and width $\maxWidth=6$. We have 
        $\gates(1)=\{g_1,g_2\}$ and $\gates(2)=\{g_3\}$. The gate $g_2$ has $\fanin=4$ and $\fanout=1$ giving $\fantotal=5$, while the gate in$_1$ has $\fanout=2$ and $\fantotal=2$. The set $\wires(0)$ and $\wires(1)$ are indicated on the figure.}
        \label{fig:circuit}
    \end{SCfigure}
     
	\subsection{Error Correcting Codes and Locally Decodable Codes}
    \label{sec:prelim:LDC}

    For an alphabet~$\Sigma$, the Hamming distance of two strings $x,y\in (\Sigma \cup \bot)^*$ of the same length, i.e., $|x|=|y|$,  is the number of indices for which $x$ and $y$ differ and is denoted by $\Hamm(x,y)={|\{ i \mid x[i] \ne y[i]\}|}$. For two strings $x \in \Sigma^*$, $y \in (\Sigma \cup \{\bot\})^*$ and value $c\in (0,1)$, we say that $y$ can be obtained by a $c$-fraction of erasures from $x$ if $|x| = |y|$, and for all $i \in [|x|]$, it holds that either $x[i] = y[i]$ or $y[i] = \bot$, where the latter case happens at most $c|x|$ times. An index $i$ in which $y[i] = \bot$ is called an erasure.
    
    For a prime power $p$, we denote by~$\Fp$ the finite field of size~$p$. An \emph{error correcting code} is a mapping $\mathsf{Enc}: \Fp^K \to \Fp^N$ that takes $K$ symbols of the alphabet $\Fp$ into $N$ symbols of the alphabet~$\Fp$.\footnote{We can map $[p]$ and $\Fp$ with a fixed isomorphism, so that $\LDCEnc:[p]^K \to [p]^N$. }
    The value $N$ is called the  block length of the code. The ratio $K/N$ is called the \emph{rate} of the code. The \emph{relative distance} of a code is the normalized Hamming distance between any two codewords, denoted $\delta=\min_{ m\ne m'} \tfrac1N {\Hamm(\mathsf{Enc}(m),\mathsf{Enc}(m'))}$. 

    Next, we formally define the notion of locally decodable codes. For the purposes of our work, we only consider the \emph{erasure} setting, in which we are given access to a possibly corrupted codeword~$y$, obtained by erasing at most $\delta$-fraction of some  $\LDCEnc(x)$ for some $x \in \Fp^K$, and an index $i \in [K]$, and our goal is to find the $i$-th symbol of $x$.

	\begin{definition}[Locally Decodable Codes (LDCs) for erasures]
        \label{def:LDC}
		An error correcting code $\LDCEnc: \Fp^K \rightarrow \Fp^N$ is said to be a $(q,\delta,\epsilon)$-LDC  if there exists a randomized decoding algorithm $\LDCDec$ that receives as input a string $y\in (\Fp \cup \{\bot\})^N$ and an index $i\in[K]$, performs at most $q$ queries to $y$, and outputs a value with the following guarantee: 
        if there exists $x \in \Fp^K$ such that $y$ can be obtained by at most a $\delta$-fraction of erasures from $\LDCEnc(x)$, then $\Pr(\LDCDec(y,i) = x[i]) \geq 1-\epsilon$.
	\end{definition}

       \begin{definition}[Non-adaptiveness]
        \label{def:non_adaptivity}
            An LDC is called \emph{non-adaptive} if for every call to its decoding algorithm~$\LDCDec$, the set of queries it performs given input index~$i$ is only a function of the randomness and the index~$i$. In particular, a query does not depend on the outcome of previous queries.
        \end{definition}
        We can think of the decoding algorithm of a non-adaptive LDC code~$\LDCDec(\cdot,i)$ as an algorithm with oracle access to the codeword, that first generates $q$~indices to query, and then, once provided these (possibly corrupt) $q$~symbols, returns the decoded message symbol.

    The following smoothness property of LDCs means that decoding an index requires querying the codeword in a ``smooth'' (near-uniform) way. This property is important in order to avoid congestion when a node decodes multiple values.  
    \begin{definition}[Smoothness]
        \label{def:smoothness}
            An LDC is called \emph{$s$-smooth} if there exists a decoding algorithm~$\LDCDec$, such that during any call to~$\LDCDec$, any entry $j \in [N]$ of the codeword is queried with probability at most~$s$.
    \end{definition}

    The following theorem suggests that smoothness is an inherent property of LDCs, since any decoding algorithm can be transformed into a smooth one.
        \begin{theorem}[\protect{\cite[Theorem~1]{KT00}}]\label{thm:allLDCsmooth}
        Every $(q,\delta,\epsilon)$-LDC of block length~$N$ is $q/\delta N$-smooth.
        \end{theorem}

\section{Computing a Circuit in the presence of crashes}
\label{sec:Simulation}

We show how to efficiently and deterministically compute a specified circuit~$C$ in the \cliquefull model, in the presence of up to $\alpha n$ crashes.  
We start, in \Cref{sec:primitives}, by describing the procedures $\Store$, $\Retrieve$, and $\BulkRetrieve$ used to store and retrieve information in our algorithm.
In section \Cref{sec:allocate} we describe another procedure, 
$\Allocate$, that assigns gates to nodes in a balanced-manner.
Finally, in \Cref{sec:TheAlgorithm}, we describe our circuit computation algorithm based on these procedures.

\subsection{\texorpdfstring{The $\Store$, $\Retrieve$, and $\BulkRetrieve$ Procedures}{The Store, Retrieve and BulkRetrieve Procedures}}
\label{sec:primitives}

As mentioned above, any information that may get lost due to node crash is stored in the network via an LDC. We first describe the LDC we use and then detail the store and retrieve procedures.

\paragraph*{The LDC instantiation.}
We assume a fixed LDC code, described in \Cref{sec:derandom} (\Cref{lem:LDC_muller_augmented}). 
Specifically, for some power of prime $q=2^{O(\sqrt{\log n})}$ and $\rho$ such that $\rho^{-1}=2^{O(\sqrt{\log n}\log \log n)}$, our LDC has an encoding function $\LDCEnc: [q]^{\rho n} \to [q]^{n}$ and a decoding function~$\LDCDec$ which, given an input index $i \in [n]$, smoothly queries $q$ indices of the codeword and decodes correctly even if up to  $\delta$-fraction of the $q$ queried symbols are erased, for some predetermined constant distance~$\alpha<\delta<1$. 
We assume that $n=q^r$ for some integer~$r$ (specified in the detailed construction);
this assumption can be lifted using standard methods. 
Note that \Cref{def:LDC} says that $\LDCDec$ is randomized, but our goal is to compute the circuit deterministically. We thus treat any randomness used by $\LDCDec$ as originating from some \emph{randomness string}, but our implementation of obtaining such randomness strings will be \emph{deterministic} rather than randomized which will render our implementation of $\LDCDec$, and hence our circuit computation algorithm, deterministic.

\paragraph*{The $\Store$ Procedure.} 
The $\Store$ procedure ``saves'' information in the network in a robust way, by encoding it with an LDC and distributing the codeword among the nodes of the network.

Each node $v_j$ begins the \Store procedure with a bit-string~$U_j$ it wishes to store in the network. It first splits $U_j$ into parts of size~$\rho n \lfloor \log{q} \rfloor$ bits each (so that each can be represented by a string of~$\rho n$ symbols over the alphabet~$[q]$), padding with zeros as necessary. 
Set $\last_j=\lceil |U_j|/(\rho n \lfloor \log{q} \rfloor)\rceil$ and denote these parts $U_j^1,\dots,U_j^{\last_j}$. 
The node~$v_j$ then encodes each $U^i_j$ using $\LDCEnc$ to obtain a codeword $\LDCEnc(U^i_j) = L^i_j$ of size $|L^i_j| = n$ symbols.  

Next, $v_j$ distributes the codewords to the network nodes. Specifically, in round $i = 1,\dots,\last_j$, it sends
the symbols of~$L_j^i$---one symbol to each node in the network. 
This takes a single round of communication because each symbol in the LDC codeword comes from the alphabet~$[q]$ with $q=2^{O(\sqrt{\log n})}$, hence it can be encoded in $\log q=O(\log n)$ bits. 
The formal description is depicted in~\Cref{alg:store}.

\begin{algorithm}[htp]
\caption{\Store ~(for node $v_j$)}\label{alg:store}
\begin{algorithmic}[1]
\Statex {\textbf{Input:}}
A bit-string~$U_j$.
\Statex
\State Partition $U_j$ into consecutive parts of size~$\rho n \lfloor \log{q} \rfloor$ bits, padding the last part if necessary; denote these substrings $U_j^1,\dots,U_j^{\last_j}$.  \Comment{$\last_j=\lceil |U_j|/(\rho n\lfloor \log{q} \rfloor)\rceil$}
\For{$i=1,\ldots,\last_j$}
    \State $L_j^i \gets \LDCEnc(U_i^j)$.
    \State For all $t\in[n]$ in parallel, send~$L_j^i[t]$ to~$v_t$.  
\EndFor
\label{line:store:record_alive}

\end{algorithmic}
\end{algorithm}

We say that $v_j$ \emph{stored}~$U_j$ in the network if $v_j$ completed the 
$\Store$ procedure without crashing.
The following is straightforward from \Cref{alg:store}.
\begin{observation}
\label{lem:store_time}
    Storing a string $U_j$ takes $O(\lceil |U_j|/(\rho n\lfloor \log{q} \rfloor) \rceil)$ rounds of communication.
\end{observation}

\paragraph*{The $\Retrieve$ and $\BulkRetrieve$ Procedures.} 
In the $\Retrieve$ procedure, a node $v_j$ is given as input an index $w$ of some string $U$ that was previously stored in the network using the \Store procedure. Node $v_j$ is additionally given a string $R$ called the randomness string. One can think of this procedure as randomized, with $R$ as its randomness, but in all our invocations of $\Retrieve$, the string $R$ is set deterministically, as will be explained later.

The goal of the $\Retrieve$ procedure is to retrieve the value of the $w$-th bit of the previously stored~$U$. 
To that end, $v_j$ first identifies the codeword that contains the bit~$w$: 
recall that the $\Store$ procedure splits $U$ into parts of size $\approx \rho n \log q$ bits. Denote by~$U^i$ the respective part and by~$i'$ the index of the symbol in~$U^i$ that contains the bit-value $U[w]$ in which we are interested. 
In the following, we say ``decode $U[w]$'' to actually mean decoding the respective index~$i'$ of the possibly corrupted codeword $\LDCEnc(U^i)$ that contains the respective value.

To retrieve the value of~$U[w]$ from the (stored) codeword $\LDCEnc(U^i)$,  
the node $v_j$~executes $\LDCDec(.)$ using the randomness string~$R$ and obtains indices of random $q$ symbols of $\LDCEnc(U^i)$ needed for the decoding. 
It is possible to learn these $q$ indices in advance because the LDC is non-adaptive (\Cref{def:non_adaptivity}).
The node $v_j$ then queries the respective nodes for their stored symbols and provides $\LDCDec(.)$ with their replies.
Note that crashed nodes do not reply, which translates to erasures given to~$\LDCDec(.)$. 
Further, in hindsight, the randomness string~$R$ will be derandomized, which has the effect that all nodes know~$R$. It will follow that $v_j$ does not actually need to send any message in order to query any node; the $q$ nodes will know that they are the nodes that should give information back to~$v_j$, since they will know the identity of~$i',U^i$ and the value of~$R$ to begin with. See \Cref{lem:bookkeeping}.

Under some circumstances, we allow the retrieval to fail, in which case the output is $\bot$. 
This could happen in two cases: (i) when there are too many erasures and $\LDCDec(.)$ returns $\bot$, or  
(ii) if one of the nodes queried during this $\Retrieve$ invocation crashes during the execution of this~$\Retrieve$ invocation. 
For the former, our derandomization (\Cref{sec:derandom}) will guarantee that this event cannot happen if at most $\alpha n$ nodes have crashed. 
For the latter, in this case we set the output to be~$\bot$ even if the decoding $\LDCDec$ successfully retrieves the symbol. 
This decision does not affect the correctness of the algorithm, but rather simplifies its analysis.

\begin{algorithm}[ht]
\caption{\Retrieve ~(for node $v_j$)}\label{alg:retrieve}
\begin{algorithmic}[1]
\Statex {\textbf{Inputs:}}
An index~$w$ to some string $U$ previously stored in the network via the $\Store$ procedure and a randomness string~$R$.

\Statex

\State Identify the part $U^i$ used by $\Store$ to encode $U[w]$ and the respective index~$i'$ in it that contains that value, i.e., the symbol~$U^i[i']$ contains the bit~$U[w]$.  
\State 
Execute $\LDCDec(\cdot,i')$ with randomness string~$R$, to obtain the $q$~indices  needed to decode $U^i[i']$.  Set $S\subset V$ to be the nodes that hold these respective symbols.

\State Query the respective nodes in~$S$. Treat symbols associated with crashed nodes as erasures.
\label{line:retrieve:query}

\Statex
\Statex \textbf{Output:} 
\State
If some node $v\in S$ crashes \textbf{during the execution of Line~\ref{line:retrieve:query}}, output $\bot$. 
\State Otherwise, output the bit $U[w]$ contained in  $\LDCDec(\LDCEnc(U^i,i'))$ by providing the $q$ replies (including erasures) when $\LDCDec$ queries the codeword~$\LDCEnc(U^i)$.
\end{algorithmic}
\end{algorithm}

\medskip

The $\BulkRetrieve$ procedure generalizes the above to allow retrieving multiple previously stored bits.
Now, $v_j$ is given as input a \emph{collection} of indices~$W_j$, where each index $w\in W_j$ refers to some (predetermined)~$U_{(w)}$ that was previously stored in the network via a~$\Store$. The strings $U_{(w)}$ may be different for different values of $w$, or they may be the same. 
Additionally, the procedure gets a multiplicity parameter~$\ell$. The goal is to output the value of the bit in index $w$ of~$U_{(w)}$, for each index $w\in W_j$. 

Towards this goal, 
all nodes in the network first \emph{deterministically} compute a set of randomness strings $\mathcal{R}(v_j) = \{R_{v_j,w,i} \mid i \in [2^\ell],w \in W_j\}$ for each $v_j \in V$, with \emph{good} properties, which we define and discuss in detail later in the section (see \Cref{def:good_seed}).
The deterministic generation of these strings is given by \Cref{lem:good_seeds_computation}.
After this step, all nodes $v\in V$ know $\mathcal{R}(v_j)$ for every~$v_j$.

Next, $v_j$ performs, for each index $w \in W_j$, a batch of $2^{\ell}$ 
$\Retrieve$ procedures, where the $i$-th invocation uses randomness string $R_{v_j,w,i}$. 
Similar to the case of $\Retrieve$, we allow some retrieves to fail and output~$\bot$.
If at least one of the $2^{\ell}$ invocations of $\Retrieve(w,R_{v_j,w,i})$ succeeds, its output becomes the output of $\BulkRetrieve$ for the index~$w$; otherwise, the respective output is~$\bot$. 

All the  $|W_j|\cdot 2^{\ell}$ \Retrieve invocations are executed in parallel. However, in order to avoid congestion, $v_j$~pipelines requests targeted to the same node. That is, it sends at most one query to any node in any given round. 
Similar to above, the set of strings $\{U_{(w)}\}_{w\in W_j}$ is predetermined and known to all nodes, 
and the identities of the $q$ nodes that are queried in a specific  $\Retrieve(w,R_{v_j,w,i})$ are generated using $R_{v_j,w,i}$, which is also known to all nodes. 
Hence, each queried node can infer the respective $U_{(w)}$ for each query, without the need for $v_j$ to communicate this data. 
The formal procedure is depicted in \Cref{alg:bulkretrieve}.

\begin{algorithm}[htp]
\caption{$\BulkRetrieve$ (for node $v_j$)}\label{alg:bulkretrieve}
\begin{algorithmic}[1]
\Statex {\textbf{Inputs:}}
An ordered collection $W_j$ of indices, where each $w\in W_j$ refers to a predetermined string~$U_{(w)}$, that  was previously stored by some node executing~$\Store$. 
A multiplicity parameter~$\ell$.

\Statex
\State Compute a collection of randomness strings $\mathcal{R}(v_j) = \{R_{v_j,w,i} \mid w \in W_j, i \in [2^\ell]\}$ using \Cref{lem:good_seeds_computation}.
\ParFor{each $w \in W_j$} \label{line:bulk:foreachW}
\ParFor{$i = 1,\ldots,2^{\ell}$ times} \label{line:bulk:foreachEll2}
\State \parbox[t]{\textwidth-3\dimexpr\algorithmicindent}{Run $\Retrieve(w,R_{v_j,w,i})$. 
If $v_j$ needs to query the same node multiple times in all of these parallel instances, send at most one query per round, until all queries are sent.\strut }
\label{line:BulkRet:retrieve}
\EndParFor
\EndParFor
\Statex
\State \textbf{Output:} 
For each $w$, the output is the output of the first successful $\Retrieve(w,R_{v_j,w,i})$, if any, or~$\bot$ otherwise.
\end{algorithmic}
\end{algorithm}

Consider a specific instance of $\BulkRetrieve(W_j,\ell)$.
The selection of randomness strings~$\mathcal{R}(v_j)$ that $v_j$ uses has a tremendous effect on the induced congestion.
Indeed, assume that $\mathcal{R}(v_j)$ is such that \emph{all} $\Retrieve$s query some $v_i$, implying $2^\ell |W_j|$ rounds of communication where in each round a single LDC symbol is communicated. 
This should be contrasted with the fully randomized case, where each node is queried in a near-uniform distribution (implied by the smoothness of LDC codes, \Cref{thm:allLDCsmooth}), implying that each $v_i$ is queried $2^\ell |W_j|\cdot q/n$ times, in expectation. Standard tail bounds show that the number of queries of the maximal node  (and hence the round complexity) is bounded by 
$2^\ell |W_j|\cdot q/n\cdot O(\log n)$. This gives that there \emph{exists} a way to select the randomness strings~$\mathcal{R}(v_j)$ while maintaining the same round complexity.

However, while the above gives uniform query locations for the goal of controlling congestion, it does not address the problem that many locations might be erased. 
To illustrate this point, assume that out of the $2^\ell$ $\Retrieve$ instances, all but one are querying mostly erased symbols (crashed nodes), and only one $\Retrieve$ correctly decodes the value. The output of $\BulkRetrieve$ would be correct in this case, but in this scenario the adversary needs to crash only a single additional node in order to fail it. 

Indeed, what we show is even stronger than mimicking a fully randomized case by some na\"ive load balancing. The following \Cref{lem:good_seeds_computation} shows that we can find a set of randomness strings that maintains a similar round complexity even if each codeword has  $\alpha n$ symbols erased, and furthermore, \emph{each individual $\Retrieve$ succeeds decoding the respective value}, as long as no new crashes happened during that $\Retrieve$.
In other words, we can de-randomize the random sampling of codeword symbols to query while (i) maintaining complexity (by controlling congestion) and (ii) performing only ``useful'' queries, hence maintaining our resilience to an all-knowledgeable adversary. 
Our choice of randomness strings guarantees that the adversary must waste $2^\ell$ of its crashing budget in order to fail the $\BulkRetrieve$, which is crucial for the correctness proof.

We now define the notion of good randomness strings, namely, strings that provide the above properties for the $\BulkRetrieve$ procedure.
\begin{restatable}[Good randomness strings]{definition}{GoodSeedDef}
\label{def:good_seed}
Fix a node $v_j$, parameters 
$W_j,\ell$, and the set of non-crashed nodes~$\alive$.
A collection of randomness strings $\mathcal{R}(v_j) = \{R_{v_j,w,i}\}_{w \in W_j, i \in [2^\ell]}$ 
is called \emph{good} for an instance of $\BulkRetrieve(W_j,\ell)$,
if the following holds: 
    \begin{enumerate}[(1)]
        \item \label{item:good_seed_cong} Each node is queried at most $O(\lceil \frac{2^\ell |W_j| q}{n} \rceil \log{n})$ times in total by $v_j$, 
        and
        \item \label{item:good_seed_correctness} 
        For all $w \in W_j$ and $i \in [2^\ell]$,  
        the invocation of $\Retrieve(w, R_{v_j,w,i})$ succeeds (given no further changes in~$\alive$). 
    \end{enumerate} 
\end{restatable}
In \Cref{sec:derandom} we prove the following.
\begin{restatable}{lemma}{GoodQueryCompute}
\label{lem:good_seeds_computation}
    There is a zero-round deterministic algorithm which, given the set of non-crashed nodes~$\alive$, a collection of indices~$W_j$, and a multiplicity parameter~$\ell$, 
    computes a good collection of randomness strings $\mathcal{R}(v_j)$ for~$v_j$. 
\end{restatable}

With the above, the following is immediate.

\begin{lemma}
\label{lem:nodeLoad}
$\BulkRetrieve(W_j,\ell)$  takes $O(\lceil \frac{2^\ell |W_j| q}{n} \rceil \log{n})$ rounds. 
\end{lemma}

\subsection{\texorpdfstring{The $\Allocate$  Procedure}{The Allocate Procedure}}
\label{sec:allocate}
The purpose of the $\Allocate$ procedure is to assign a set of given gates that need to be computed to non-crashed nodes. 
The procedure is deterministic and runs locally on each node, without any communication. 
However, all nodes reach the same allocation, since they all have the same knowledge regarding crashed nodes and regarding failed LDC queries, where the latter is due to the randomness strings being generated deterministically by all nodes and known to all (due to \Cref{lem:good_seeds_computation} above and the upcoming \Cref{lem:bookkeeping} which essentially says that the nodes are able to keep a consistent view of all of these variables by careful bookkeeping).

In more detail, the procedure 
is given as input the current set of non-crashed nodes~$\alive$, a set of gates~$G\subseteq V_C$ (of the circuit $C=(V_C,E_C)$, known to all), and a multiplicity parameter~$\ell_1$.
For each gate $g \in G$, $\Allocate$ assigns $g$ to a set of $\min(2^{\ell_1},|\alive|)$ nodes from~$\alive$ using the following  sequential ``greedy'' process: 
Sort the gates in~$G$ by their total fan (denoted $\fantotal$), in descending order. 
Assign the gates one by one to a set of $\min(2^{\ell_1},|\alive|)$ distinct nodes in~$\alive$ whose loads are minimal (break symmetry by node IDs).
The \emph{load} of a node $v$, denoted $\lambda(v)$, is defined as the sum of the total $\fantotal$ of all  gates assigned to it so far during this $\Allocate$ instance. 
See \Cref{alg:allocate}. 

\begin{algorithm}
\caption{\Allocate ~(for node $v_j$)}\label{alg:allocate}
\begin{algorithmic}[1]
\Statex {\textbf{Inputs:}} 
A set~$G$ of gates and a multiplicity parameter~$\ell_1$.

\Statex
\State Let $g_1,\ldots,g_{|G|}$ be the gates of~$G$, sorted in descending order of~$\fantotal$ (break ties consistently).
\State Set $\lambda(v) \gets 0$ for all nodes $v \in \alive$.
\State $G_j \gets \emptyset$.
\State $L \gets \min(2^{\ell_1},|\alive|)$. 
\For{$i = 1,\ldots,|G|$}
\State Let $u_1,\ldots,u_L$ be the $L$ nodes in~$\alive$ with the minimal loads (break ties by IDs).
    \ForEach {node $u \in \{u_1,\dots,u_L\}$}
      \State Assign~$g_i$ to the node~$u$.
      \State $\lambda(u) \gets \lambda(u)+\fantotal(g_i)$.
      \If{$u =v_j$}
        $G_j \gets G_j \cup \{g_i\}$.
      \EndIf
    \EndFor
\EndFor
\Statex

\State \textbf{Output:} The set $G_j$ of gates assigned to~$v_j$.
\end{algorithmic}
\end{algorithm}

The assignment can be computed locally in a consistent manner across all non-crashed nodes without any communication, since all relevant information, namely $G,\ell_1$, and~$\alive$, is known to all nodes in~$\alive$.
Note that since this is a local computation procedure, we can assume that set $\alive$ does not change throughout the computation, and that all nodes use the same set~$\alive$ representing the non-crashed nodes at the beginning of that round.

The following lemma bounds the load assigned to each node.
\begin{restatable}{lemma}{wireLoadBalanceLem}
\label{lem:alloc_load_balance}
Let $L=\min(2^{\ell_1},|\alive|)$ and $P=\sum_{g\in G}\fantotal(g)$. Further, assume $\max_{g\in G} \fantotal(g) \le \Delta$.
Then, $\Allocate(G,\ell_1)$ puts a maximal load of $\max(4PL/|\alive|,\Delta)$ on each node.
\end{restatable}
\begin{proof}
    Set $r = |G|$. let $g_1,\dots,g_r$ be the gates in~$G$ sorted in a decreasing order of their~$\fantotal$.
    For any $i\in[r]$, set $d_i = \fantotal(g_i)$, then,
    \[
    d_r \leq d_{r-1} \leq \ldots \leq d_1 \leq \Delta\text{.}
    \]

    First, we analyze the case where $|\alive|/2 \leq L \leq |\alive|$. We notice that the load of each vertex~$v$ is trivially bounded by $\lambda(v) \leq \sum_{i=1}^r d_i$ (which is obtained if all gates are allocated to~$v$). On the other hand, by assumption that $|\alive|/2 \leq L$ it follows that
    \[
    \lambda(v) \leq \sum_{i=1}^r d_i = P \leq 2PL/|\alive|\text{,}
    \]
    and the claim for the case of $|\alive|/2 \leq L \leq |\alive|$ follows. 
    
    We assume in the remainder of the proof that $L < |\alive|/2$. First, we make the very simple observation that at all times during the procedure, it holds that
    \begin{equation}
    \label{eq:alloc_max_total_load}
    \sum_{v \in V} \lambda(v) \leq \sum_{i=1}^r d_i \cdot L = PL.
    \end{equation}
    
    We split the analysis into two phases of the greedy allocation procedure: we say that the procedure is in its first phase while $d_i \geq 2PL/|\alive|$ and we say that it is in its second phase while $d_i < 2PL/|\alive|$. In each phase, we show that following a gate being allocated to batch of $L$ vertices, all vertices have load at most $\max(4PL/|\alive|,\Delta)$.

    While $d_i \geq 2PL/|\alive|$, we must have at least $|\alive|/2$ vertices with no load. Assume otherwise, then we notice that since the gates $g_1,\dots,g_r$ are sorted in descending order according to $d_1,\dots,d_r$, then each vertex with an allocated gate has load at least $d_i \geq 2PL/|\alive|$. Summing the loads of all vertices, we conclude that the total load is at least $\frac{2PL}{|\alive|}(\frac{|\alive|}{2}+1)  > PL$, contradicting \Cref{eq:alloc_max_total_load}. Hence, after any allocation where $d_i \geq 2PL/|\alive|$, we only pick vertices with $\lambda(v) = 0$, and the load of any such vertex $v$ is therefore at most $\Delta$.
    
     Next, we consider the phase where $d_i \leq 2PL/|\alive|$. By an averaging argument on \Cref{eq:alloc_max_total_load}, at any point of time of the procedure we have at least $|\alive|/2$ vertices with load at most $2PL/|\alive|$. Since $L < |\alive|/2$, the new load of a vertex $v$ that is assigned a new gate is at most 
     \[
     \lambda(v) \leq d_i + \frac{2PL}{|\alive|} \leq \frac{4PL}{|\alive|},
     \]
     and the claim follows.
\end{proof}

\subsection{The Circuit Computation Algorithm}
\label{sec:TheAlgorithm}

We can now complete the description of our circuit computation algorithm, presented in \Cref{alg:main}.
The algorithm takes as input a circuit~$C$ whose inputs (the wires~$\wires(0)$) are already stored in the network.

The algorithm computes the gates of~$C$ layer by layer, in a sequence of $d$ steps referred to as \emph{$\layerStep$-steps}. For $\layerStep$-step $i=1,\ldots,d$, we assume that $\wires(0),\ldots,\wires(i-1)$ have already been stored by previous iterations, and the goal is to compute and store~$\wires(i)$. 
To this end, 
the nodes execute $\Allocate$, which assigns to each non-crashed node~$v_j$ a set of gates ~$G_j\subseteq \gates(i)$ to compute and store their output wires (line~\ref{line:main:allocate}).

Then, each node $v_j$ tries to retrieve the input wires of the gates $G_j$ assigned to it via the $\BulkRetrieve$ procedure (line~\ref{line:main:retrieve}). If successful, the node computes the gates assigned to it (line~\ref{line:main:compute}) and obtains the values of all output wires~$U_j$ of the gates~$G_j$. The node then stores these wires in the network (line~\ref{line:main:store}).

However, crashes that occur during this computation may hinder the computation of some wires.
To overcome this issue, the computation of layer~$i$ consists of two nested loops.
The outer loop, which we call the  $\outStep$-loop, iterates over 
$\ell_1 = 1,\dots,\lceil \log{n} \rceil$ and doubles the number of nodes that try to compute a given gate.
The inner loop, called the $\inStep$-loop, iterates over
$\ell_2 = 1,\dots,\lceil \log{\wireLoad} \rceil$ for some parameter $\wireLoad$ (fixed in the analysis), and doubles the number of retrieval attempts a given node performs for each input wire assigned to it.

If during the computation of $\layerStep$~$i$, more than $c_f n/(q \log n)$ 
new crashes have occurred, 
for some sufficiently small constant $c_f > 0$ determined later,
we re-start the computation of that layer with the remaining nodes. Namely, we maintain a counter $f_{i,rep}$ of newly crashed nodes in the $\layerStep$-step, which is initialized at the start of the $\layerStep$-step to be $0$. Once it passes $c_f n/(q \log n)$, we reset the counter to~0, reset the multiplicity parameters $\ell_1,\ell_2$ to~$1$, and retry to compute and store all remaining unstored wires in~$\wires(i)$. 
This action is captured in lines~\ref{line:main:rep_start}--\ref{line:main:rep_end}, assisted by the variable~$rep$, that counts the number of repetition attempts of computing layer~$i$. We call such a repetition of a $\layerStep$-step \emph{overwhelmingly faulty}:
\begin{definition}
\label{def:overfaulty_step}
Let $c_f > 0$ be a sufficiently small constant. 
A repetition of a $\layerStep$-step is called overwhelmingly faulty if $c_f n/(q \log n)$ new crashes occur during this step. See \Cref{eq:overfaulty_constant_def} in \Cref{sec:analysis} for the exact definition of $c_f$.
\end{definition}

\begin{algorithm}[ht!]
\caption{\protect Robust Circuit Computation   
(for node~$v_j$) %
}\label{alg:main}
\begin{algorithmic}[1]
\Statex {\textbf{Inputs:}}
\Statex
A globally known circuit $C$.

\Statex The inputs $\wires(0)$ to the circuit~$C$ are stored in the network via the $\Store$ procedure.
\Statex Global parameters $\wireLoad$, $\maxWaitRet$ and $\maxWaitStore$  (to be set later).

\Statex

 \For {Layer $i = 1,\ldots, d$} 
 \label{line:main:LayerLoop}
 \Comment{The $\layerStep$-loop}
 \State $rep \gets 1$, $S \gets \emptyset$
 \label{line:main:init_i_rep}
 \For {$\ell_1 = 1,\ldots,\lceil \log{n} \rceil$}\Comment{The $\outStep$-loop}\label{line:main:l1loop}
 \State $G \gets \gates(i)\setminus S$
 \label{line:main:updateG}
\State $G_j \gets \Allocate(G,\ell_1)$ 
 \label{line:main:allocate}

\State Let $W_j$ be the set of wires required for computing all gates in $G_j$.
 \For {$\ell_2 = 1,\ldots, \lceil \log{\wireLoad} \rceil$}\Comment{The $\inStep$-loop}
    \State Execute $\BulkRetrieve(W_j,\ell_2)$. 
    \label{line:main:retrieve}
    \Comment{Idle  until  $\maxWaitRet$ rounds pass}
    \State Remove from $W_j$ all wires that were successfully retrieved in line~\ref{line:main:retrieve}.

\If{$W_j = \emptyset$}
\Comment{Otherwise, idle} \label{alg:main:ifAllWjSuccess}
    \State Let $U_j$ be the output wires of $G_j$. \State Locally compute the values of~$U_j$ using the retrieved values. 
    \label{line:main:compute}
    \State Execute $\Store(U_j)$. 
    \label{line:main:store}\Comment{Idle  until $\maxWaitStore$ rounds pass }
\EndIf
\State Update $S$ to include gates whose output wires were stored by at least one~node. 
\label{line:main:updateS}

\Statex
\Comment{If too many failures have occurred, repeat   layer~$i$}
\State Let $f_{i,rep}$ be the number of crashes since the last execution of line~\ref{line:main:init_i_rep} or \ref{line:main:repeat}.%
\label{line:main:rep_start}
\If{$f_{i,rep} > c_f n/(q\log{n})$} 
\State $rep\gets rep+1$
\State\textbf{continue from} line~\ref{line:main:l1loop}, re-setting $\ell_1 \gets 1$. 
\label{line:main:repeat}
\EndIf
\label{line:main:rep_end}

\EndFor
\EndFor

\EndFor
\end{algorithmic}
\end{algorithm}

The next lemma captures the following observation: the nodes are capable of `bookkeeping' the progress of the computation at any given round of \Cref{alg:main}.
This bookkeeping information
includes gates that were computed and stored, gates that still need to be computed, $\Retrieve$ calls that succeeded and those that failed, etc.
In particular, when a node needs to access some wire~$w$, that was previously stored, the node knows exactly which LDC codeword contains it, and which index of that codeword it should decode in order to retrieve~$w$. 
\begin{lemma}[Bookkeeping]
\label{lem:bookkeeping}
Any non-crashed node knows, at the start of any round,
the following information:
(1) the set~$S$ of gates whose outputs were stored in the network, and 
(2) for any $g\in S$ and any output~$w$ of~$g$, the LDC codeword that contains $w$ (namely, the node~$v_j$ that stored it and the round in which it was stored) and the index of~$w$ in the string~$U_j$ that $v_j$ stored.
\end{lemma}
\begin{proof}
Recall that, by definition, since a crashed node does not send any messages starting from the round in which it crashes, all nodes know the set~$\alive$ of non-crashed nodes at the beginning of every round. 

We prove Items (1) and (2) by induction on the round number.
At initialization (the beginning of the first round, $r=1$), the inputs to the circuit~$C$ are assumed to be stored, hence (1) and (2) hold for the input gates~$\gates(0)$. 

Next, we assume the statement holds at the beginning of some $\outStep$-step, and we show it holds in every round until the end of this $\outStep$-step.
Note that all nodes execute \Cref{alg:main} in synchrony and, specifically, they all perform $\BulkRetrieve$ or $\Store$ at the same rounds (other actions do not involve communication as they are purely computational and thus take zero rounds).

If round $r$ is not the final round of a $\Store$ procedure, then the set of stored wires is unchanged.  
Otherwise, each node knows the set~$S$ of stored gates at the beginning of the $\Store$, by the induction hypothesis, and thus it also knows the set $G = \gates(i)\setminus S$.
Since $\Allocate$ is deterministic and depends only on $C$, $\ell_1$, $G$, and~$S$, then all nodes 
learn the same output of $\Allocate(G,\ell_1)$. 
In particular, they all learn  
the gates~$G_j$ that each $v_j\in \alive$ is assigned to compute in this $\outStep$-step.

With this knowledge, all nodes 
can (locally) generate the good randomness strings that are used by some~$v_j$ for each of its $\BulkRetrieve$ invocations (\Cref{lem:good_seeds_computation}). 
Further, all nodes have the same knowledge 
about $\Retrieve$ calls that failed in previous rounds due to new crashes. They can thus infer which nodes $v_j$ have successfully retrieved all input wires of $G_j$ (i.e., those for which $W_j = \emptyset$) and satisfy the condition of in \Cref{alg:main:ifAllWjSuccess}. Only these nodes perform the $\Store$ that completes in that round~$r$. 

Out of the nodes that perform $\Store$, any node~$v_j$ that 
does not crash before round~$r$,  succeeds in storing all the wires in~$U_j$ (i.e., all the output wires of~$G_j$).

Therefore, at the start of round $r+1$,  all the nodes in~$\alive$ learn the set of nodes that performed a  successful $\Store$. They also know the set $U_j$ of each $v_j$ that completed a $\Store$, in particular, which wires it contains and their internal order.\footnote{While $U_j$ is a set of wire-values, the $\Store$ procedure expects a bitstring, hence we induce some standard order on the set~$U_j$ that maps it into a bitstring, and this ordering is known by all nodes.} 
This implies Item~(2).
It also follows that all nodes in~$\alive$ can update their set~$S$ of stored gates in a consistent manner (line~\ref{line:main:updateS}), which implies Item~(1).
\end{proof}

As a result of this careful bookkeeping, $v_j$ does not  need to send any ``metadata'' information to the nodes it needs to query---they already have all the needed information (in the notations of $\BulkRetrieve$, they know $w$ and $U_{(w)}$, the randomness strings~$\mathcal{R}(v_j)$, and the specific round(s) in which $v_j$ queries them (i.e., is expecting a symbol from them). 

\section{Analysis}
\label{sec:analysis}
\medskip

In this section we prove that our circuit computation algorithm satisfies the requirements of our main theorem, restated here for convenience. 
\MainWireThm*

We prove the main bulk of \Cref{thm:circuit} via the following  \Cref{lem:layer_step_main}, stating that  once \Cref{alg:main} completes its $i$-th \layerStep-step (i.e., the $i$-th iteration of the loop in line~\ref{line:main:LayerLoop}),
the wires of the $i$-th layer,  $\wires(i)$,  are successfully stored in the network. This immediately leads to the correctness of the algorithm, since after $d$ \layerStep-steps, the algorithm completes computing and storing all the wires of the circuit~$C$, including all of its output wires.

\begin{lemma}
\label{lem:layer_step_main}
   After the $i$-th $\layerStep$-step ends,
   all $\wires(i)$  are stored in the network.
\end{lemma}

The complexity of the algorithm, as stated in \Cref{thm:circuit}, 
is proved via the following lemma, whose proof can be found in 
\Cref{subsec:rounds}.

\begin{restatable}{lemma}{LemmaRounds}
\label{lem:round_complexity_main}
    The round complexity of 
    \Cref{alg:main}
    is $d \lceil\maxWidth/n^2+\Delta/n \rceil)2^{O(\sqrt{\log{n}}\log\log{n})}$.
\end{restatable}

The above two lemmas allow us to prove our  main theorem.
\begin{proof}[Proof of \Cref{thm:circuit}]
    By \Cref{lem:layer_step_main}, at the end of 
    $\layerStep$-step $i$, the wires $\wires(i)$ are stored in the network. Hence, after the last $\layerStep$-step,  i.e.\@ layer~$d$, 
    all $\wires(i)$ for $i\le d$ are stored in the network.
    These include all the outputs of the circuit~$C$, which proves the correctness of the algorithm.
    By \Cref{lem:round_complexity_main}, the round complexity of the algorithm is $d \lceil\maxWidth/n^2+\Delta/n \rceil)2^{O(\sqrt{\log{n}}\log\log{n})}$. 
\end{proof}

The main technical effort is in proving \Cref{lem:layer_step_main}. 
The crux of the argument relies on showing that at the end of each $\outStep$-step, if there were not too many new crashes, then each non-crashed node successfully stores all its allocated wires. 
Recall that a certain repetition of a $\layerStep$-step is called 
overwhelmingly faulty if there are more than $c_f n / (q\log n)$ newly crashed nodes during that repetition (\Cref{def:overfaulty_step}).

\begin{lemma}
  \label{lem:l2iteration_main}
      At the end of any $\outStep$-step in a non-overwhelmingly faulty repetition of a $\layerStep$-step, the following holds for all $j \in [n]$: if $v_j$ is not crashed, then $v_j$ succeeds in storing the all the outputs of all the gates in~$G_j$, namely, it succeed in storing~$U_j$.
  \end{lemma}

Given \Cref{lem:l2iteration_main}, \Cref{lem:layer_step_main} follows immediately:
\begin{proof}[Proof of \Cref{lem:layer_step_main}]
    By \Cref{lem:l2iteration_main}, at the end of a $\outStep$-step every non-crashed node $v_j$ stores $U_j$. Moreover, we notice that for $\ell_1 = \lceil \log{n} \rceil$, each unstored wire of $\wires(i)$ is allocated to all non-crashed nodes. Since the total number of node crashes throughout the algorithm is at most $\alpha n < n$, it follows that at least one node is non-crashed after this step ends, and thus all remaining wires are stored following this step.
\end{proof}

Our primary goal in this section is therefore to prove \Cref{lem:l2iteration_main}. 
Our path is as follows. We prove the following using induction: (a) the number of wires belonging to gates allocated to each node at the start of a $\outStep$-step is not too large, specifically, not larger than $\wireLoad = \Theta(\max(\maxWidth/n,\Delta,n))$, and (b) at each $\inStep$-step, some upper bound on the number of wires $|W_j|$ that node $v_j$ still needs to retrieve in order to compute its allocated gates~$G_j$ decreases by at least half (with respect to the previous iteration). These will lead to the lemma.

For the remainder of the section, we associate a step of the algorithm with four values $(i,rep,\ell_1,\ell_2)$, denoting that we refer to repetition $rep$ of $\layerStep$-step $i$, where the inner loops are set to $\outStep$-step $\ell_1$ and $\inStep$-step $\ell_2$, respectively.

\subsection{Proof of Lemma~\ref{lem:l2iteration_main}}

Throughout  this subsection, we fix some $\layerStep$-step and a repetition $rep$. We assume that it is \emph{not} overwhelmingly faulty (\Cref{def:overfaulty_step}): If it is, then it is restarted (a new repetition), and we do not account for any progress during this repetition. Nevertheless, 
the remaining budget of failures allowed in future repetitions
significantly decreases.

     Recall that $\maxWidth $ and $ \Delta$ denote the width of~$C$ and the maximum fan of a gate, respectively, and that $\alpha$ is the fraction of crashes allowed for the adversary. 
     We set 
     \begin{equation}
     \wireLoad = \max\left(\frac{8\maxWidth}{(1-\alpha) n},\Delta,n\right)\text{,}
     \end{equation}
     which, informally, is a rough upper bound on the number of wires a node needs to retrieve in order to be able to compute its allocated gates, and on the number of wire it needs to store. Additionally, we set $c_f$ from \Cref{def:overfaulty_step}. Let $c_1 > 0$ be the constant of \Cref{def:good_seed}(1), then:
     \begin{equation}
     \label{eq:overfaulty_constant_def}
         c_f = \min\left(\frac{1-\alpha}{16},\frac{1}{4c_1}\right).
     \end{equation}

    The variables of the algorithm (e.g. $U_j,W_j,G_j,\alive,G,S$) change in the course of the algorithm. For any such variable $X$, we denote by $X(\ell_1,\ell_2)$ the value $X$ holds at the start of $\inStep$-step $\ell_2$ of $\outStep$-step $\ell_1$, and we set $X(\ell_1) = X(\ell_1,1)$. For example, $U_j(2)$ is defined as the value of $U_j$ at the start of the first $\inStep$-step of the second $\outStep$-step. 
    
    Denote by $\lambda_j(\ell_1) := |U_j(\ell_1)|+|W_j(\ell_1,1)|$ the load of node $v_j$ at the start of the first $\inStep$-step of $\outStep$-step $\ell_1$, i.e., the total fan of all gates allocated to node $v_j$ at the start of the corresponding $\outStep$-step.
    
    We say that a $\outStep$-step satisfies the \emph{load condition} if the total fan of all gates allocated to all nodes is bounded by $\wireLoad$. In other words, the load condition is satisfied in $\outStep$-step $\ell_1$ if and only if
    \begin{equation}
    \label{eq:load_condition}
       \max_{j \in [n]} \lambda_j(\ell_1) \leq \wireLoad. 
    \end{equation}
    A key part of our argument is to show that the load condition holds for all $\outStep$-steps, which we prove in \Cref{lem:l1invariant}.

    The following lemma shows that given that the $\outStep$-step satisfies the load condition, we can bound the number of not-yet-retrieved wires in~$|W_j(\ell_1,\ell_2)|$ for every node~$v_j$ by $\wireLoad /2^{\ell_2}$ for every $\inStep$-step $\ell_2$.
 \begin{lemma}
 \label{lem:l2iteration}
     If a $\outStep$-step in a non-overwhelmingly faulty repetition of a $\layerStep$-step satisfies the load condition (Eq.~\eqref{eq:load_condition}), then at the start of any of its $\inStep$-steps it holds that if $v_j\in \alive$, 
     then $|W_j(\ell_1,\ell_2)|\le \wireLoad /2^{\ell_2-1}$.
 \end{lemma}

 \begin{proof}
     We prove the claim by induction on $\ell_2$. We notice that the base case of $\ell_2 = 1$ follows by the load condition, which states that $\max_j \lambda_j(\ell_1) \leq \wireLoad$, and in particular $|W_j(\ell_1,1)| \leq \wireLoad$.
     
     Suppose the claim holds for some~$\ell_2$.  By the induction hypothesis, we have $|W_j(\ell_1,\ell_2)|\le \wireLoad/2^{\ell_2-1}$. By definition of the $\BulkRetrieve$ procedure (lines \ref{line:bulk:foreachW}--\ref{line:bulk:foreachEll2} of \Cref{alg:bulkretrieve}), $v_j$ performs $2^{\ell_2}|W_j(\ell_1,\ell_2)|\le 2\wireLoad$ invocations of $\Retrieve$ during this $\inStep$-step.
     
     Each $\Retrieve$ invocation that $v_j$ performs produces $q$ queries. 
     Since the randomness string we use for the $\Retrieve$ is good, it follows from  \Cref{def:good_seed}(\ref{item:good_seed_correctness}) that each \Retrieve succeeds {assuming no new crashes happen during the \Retrieve invocation on any of its queried nodes}. 
     In other words, if none of the $q$ queried nodes crashes during a $\Retrieve$ invocation to retrieve $w$, then $v_j$ learns~$w$. 
     However, some of the queried nodes may crash during a \Retrieve. Suppose that $x$ new crashes occur during any of the rounds of  \Retrieve invocations of this $\inStep$-step. As defined in Algorithm~\ref{alg:retrieve},  a (new) crash of even a single node out of the $q$ nodes that $v_j$ queries leads to the failure of that~$\Retrieve$. 
     Again, since the randomness string that is used is good  (\Cref{def:good_seed}(\ref{item:good_seed_cong})), no node is queried more than $c_1 \lceil \frac{\wireLoad q}{n} \rceil \log n$ times during this $\inStep$-step, where $c_1 > 0$ is the constant of \Cref{def:good_seed}(1). Therefore, each crashing node can fail at most $c_1 \lceil \frac{\wireLoad q}{n} \rceil \log n$ different \Retrieve invocations of node $v_j$. 
  
     It follows that in order for the adversary to make $v_j$ fail in decoding the wire $w\in W_j(\ell_1,\ell_2)$ throughout the entire $\BulkRetrieve$, the adversary must fail all the $2^{\ell_2}$ parallel calls of $\Retrieve$ relating to $w$ (line~\ref{line:BulkRet:retrieve}).  
     By a pigeonhole argument, $x$ new crashes fail at most $x \cdot c_1 \lceil \frac{\wireLoad q}{n} \rceil \log n$ different \Retrieve invocations and, 
     as a consequence, fail the decoding of a number of wires which is at most 
     \begin{equation}
     \label{eq:number_of_failed_wires}
         \frac{x \cdot c_1 \lceil \frac{\wireLoad q}{n} \rceil \log n}{2^{\ell_2}} = O\left(\frac{x}{2^{\ell_2}} \left\lceil \frac{\wireLoad q}{n}\right\rceil  \log n\right).
     \end{equation}

        Recall that at the beginning of this $\inStep$-step, we had $|W_j(\ell_1,\ell_2)| \le \wireLoad/2^{\ell_2-1}$, and our goal is to prove the induction step and show that the number of missing wires at the end of this iteration is $|W_j(\ell_1,\ell_2+1)| \le \wireLoad/2^{\ell_2}$. 

     If $|W_j(\ell_1,\ell_2)| \le \wireLoad/2^{\ell_2}$, then trivially also $|W_j(\ell_1,\ell_2+1)| \le \wireLoad/2^{\ell_2}$. 
     Thus, for the remainder of the proof, we consider the case $|W_j(\ell_1,\ell_2)| > \wireLoad/2^{\ell_2}$. To prove the claim,  we require that the number of wires that $v_j$ retrieves in this step is at least~$|W_j(\ell_1,\ell_2)|/2$.  
     Thus, the $\inStep$-step is successful if the adversary fails
     less than~$|W_j(\ell_1,\ell_2)|/2$ wires.
     
     Yet, in order to fail $|W_j(\ell_1,\ell_2)|/2$ wires the adversary must crash at least $x>\frac{n}{4c_1q\log{n}}$ new nodes. This follows from \Cref{eq:number_of_failed_wires}, which bounds the number of wire failures stemming from~$x$ new crashes. We require,
     \[
     \frac{|W_j(\ell_1,\ell_2)|}{2} \le \frac{x \cdot c_1 \lceil \frac{\wireLoad q}{n} \rceil \log n}{2^{\ell_2}}.
     \]
     Recall that we are in the case where $|W_j(\ell_1,\ell_2)| > \wireLoad/2^{\ell_2}$, thus, we obtain 
     \[
     x \ge \frac{2^{\ell_2} \wireLoad}{2\cdot 2^{\ell_2}} \cdot \frac{ 1}{c_1 \lceil \frac{\wireLoad q}{n} \rceil \log{n}} >\frac{ n}{4 c_1 q \log n} \geq \frac{c_f}{q\log{n}}n,
     \]
     where $c_f$ is the constant of \Cref{def:overfaulty_step}, and the last inequality follows by \Cref{eq:overfaulty_constant_def}. 
     However, if the adversary fails $x> \frac{c_f}{q\log n} n$ new nodes during the execution of $\BulkRetrieve(W_j,\ell_2)$,
     the respective $\layerStep$-step is overwhelmingly faulty (\Cref{def:overfaulty_step}), which is a contradiction to the premise that this repetition of the $\layerStep$-step is not overwhelmingly faulty. 
     This completes the proof since the adversary cannot fail $|W_j(\ell_1,\ell_2)|/2$ wires, and thus $|W_j(\ell_1,\ell_2+1)| \le \wireLoad/2^{\ell_2}$ and the induction step holds.
 \end{proof}

    Next, we conclude that any non-crashed $v_j$ succeeds in storing $U_j$.
  
  \begin{corollary}
  \label{cor:l1iteration}
      For a $\outStep$-step in a non-overwhelmingly faulty repetition of a $\layerStep$-step that satisfies the load condition (Eq.~\eqref{eq:load_condition}), the following holds for all $j \in [n]$: if $v_j$ is not crashed at the end of the $\outStep$-step, then $v_j$ has succeeded in storing its~$U_j$.  
  \end{corollary}
    \begin{proof}
          By \Cref{lem:l2iteration}, after the final $\inStep$-step with $\ell_2=\lceil \log \wireLoad \rceil$, the node $v_j$ has $|W_j| \le \wireLoad/2^{\lceil \log{\wireLoad} \rceil} <1$. 
          It follows that $W_j = \emptyset$. Namely, node $v_j$ computed and stored all wires in $U_j$, and the claim follows. 
    \end{proof}

Next, we prove that every $\outStep$-step satisfies the load condition, which is required for \Cref{lem:l2iteration}. Recall that $\maxWidth$ denotes the width of the circuit, $\Delta$ denotes the maximum fan of a gate, and $\alive(\ell_1)$ denotes the set of non-crashed nodes at the start of $\outStep$-step $\ell_1$.
    
\begin{lemma}
\label{lem:l1invariant}
     Every $\outStep$-step in a non-overwhelmingly faulty repetition of a $\layerStep$-step satisfies the load condition in Eq.~\eqref{eq:load_condition}. 
\end{lemma}
\begin{proof}
    We prove the claim by induction on $\ell_1$. For $\ell_1 = 1$, the claim follows by the fact that the total fan of each layer is bounded by~$\maxWidth$ and the total fan of a single gate is bounded by~$\Delta$.
    Hence, by \Cref{lem:alloc_load_balance}, the total load of each node $v_j$ is at most 
    \[\lambda_j(1) \leq \max \left(\frac{4 \cdot 2^1 \cdot \maxWidth}{|\alive(1)|},\Delta\right) = \max \left(\frac{8\maxWidth}{|\alive(1)|},\Delta \right) \leq \wireLoad.\]

    Now, assume that $\max_{j \in [n]}\lambda_j(\ell_1) \le \wireLoad$ at the start of $\outStep$-step $\ell_1$ for some $1 \leq \ell_1 < \lceil \log{n} \rceil$. By \Cref{cor:l1iteration}, every non-crashed node $v_j$ successfully stores its allocated wires~$U_j(\ell_1)$ by the end of the final $\inStep$-step.
    Since the repetition of the $\layerStep$-step is not overwhelmingly faulty, then all but $\frac{c_fn}{q\log{n}}$ nodes compute and store their allocated wires $U_j(\ell_1)$, where $c_f$ is the constant of \Cref{def:overfaulty_step}. This is less than  $\frac{(1-\alpha)n}{8}$, due to the  definition of $c_f$ (\Cref{eq:overfaulty_constant_def}). 
 
     Any yet-unstored wire $w$ is such that $w$ was allocated to exactly $\min(2^{\ell_1},|\alive(\ell_1)|)$ crashed nodes. If $2^{\ell_1} \geq |\alive(\ell_1)|$, then we are done: it implies that each wire in $U$ is  allocated to all non-crashed nodes, and since not all the remaining nodes may crash then by \Cref{cor:l1iteration}, the wire set $U$ is stored at the end of the step, i.e., all wires of the current layer are stored. 
    
    Otherwise, we bound the sum of total fan of all gates whose outputs are not stored at the end of the step, i.e., we bound $\sum_{g\in \gates(U(\ell_1+1))}\fantotal(g)$. 
    At most $\frac{(1-\alpha)n}{8}$ nodes from $\alive(\ell_1)$ crashed, each with load at most $\wireLoad$. Moreover, each wire is assigned to at least $2^{\ell_1}$ nodes in $\alive(\ell_1)$. It follows that
    \[
    \sum_{g\in G(\ell_1+1)}\fantotal(g) \leq \frac{(1-\alpha)n\Lambda}{8 \cdot 2^{\ell_1}}\text
    {.}
    \]
    The load of each node $v_j$ is therefore at most 
    
    \[
    \lambda_j(\ell_1+1) \leq \max\left(\frac{4\cdot 2^{\ell_1+1}(1-\alpha)n\wireLoad}{8 \cdot 2^{\ell_1}|\alive(\ell_1+1)|},\Delta\right) = \max\left(\frac{(1-\alpha)n\wireLoad}{|\alive(\ell_1+1)|},\Delta\right) \leq \max\left(\frac{(1-\alpha)n\wireLoad}{(1-\alpha)n},\Delta \right) \leq \wireLoad\text{,}
    \]
    where the first inequality follows by \Cref{lem:alloc_load_balance}, and the last inequality follows since $\Delta \leq \wireLoad$. The claim then follows. 
\end{proof}

We are now ready to prove \Cref{lem:l2iteration_main}, which concludes the correctness proof of \Cref{alg:main}.

\begin{proof}[Proof of \Cref{lem:l2iteration_main}]
    By \Cref{lem:l1invariant}, every $\outStep$-step satisfies the load condition in Eq.~\eqref{eq:load_condition}. 
    Hence, by \Cref{cor:l1iteration}, at the end of every $\outStep$-step, each non-crashed node $v_j$ stores~$U_j$.
\end{proof}

\subsection{Round Complexity}
\label{subsec:rounds}
We bound the round complexity of our circuit computation algorithm. Recall that $\wireLoad = \max\left(\tfrac{8\maxWidth}{(1-\alpha)n},\Delta,n\right)$, $q =  2^{O(\sqrt{\log n})}$ 
and  $(1/\rho) = 2^{O(\sqrt{\log{n}}\log\log{n})}$. 
We set $\maxWaitRet=O(\lceil \wireLoad q/n \rceil \log{n})$ and $\maxWaitStore=O(\lceil \wireLoad/(\rho n \log{q}) \rceil)$ for some specifically known, sufficiently large constants, and prove our claimed round complexity.

\LemmaRounds*

\begin{proof}
    Since there can be at most~$ \alpha n$ crashes throughout the computation, the number of overwhelmingly faulty repetitions throughout the algorithm is at most~$\alpha n/(c_f n/(q \log n)) = O(q\log{n})$, where $c_f$ is the constant in \Cref{def:overfaulty_step}. Therefore, there are at most $O(d+q\log{n})$
     repetitions of $\layerStep$-steps in total. Each $\layerStep$-step is comprised only of $O(\log{n} \cdot \log{\wireLoad})$ invocations of $\Store$ and $\BulkRetrieve$ (and local computation). 

    We claim that every invocation of $\Store$ takes at most $\maxWaitStore$ rounds. To see this, fix an iteration associated with $(i,rep,\ell_1,\ell_2)$. By \Cref{lem:l1invariant}, it holds that $|U_j(\ell_1)| \leq \lambda_j(\ell_1) \leq \wireLoad$. Therefore, by \Cref{lem:store_time}, the $\Store$ completes after $O(\lceil \wireLoad /(\rho n \log{q}) \rceil)$ rounds, which is the value of $\maxWaitStore$.

    In addition, we claim that every invocation of $\BulkRetrieve$ takes at most $\maxWaitRet$ rounds. To see this, note that by combining \Cref{lem:l1invariant} and \Cref{lem:l2iteration}, it holds that $|W_j| \leq \wireLoad/2^{\ell_2-1}$.  Therefore, by \Cref{lem:nodeLoad}, the round complexity is at most $O(\lceil \wireLoad q/n \rceil \log{n})$ times, which is the value of $\maxWaitRet$.

    Recall we assume that all the circuit parameters are polynomial in~$n$.
    Thus, $\log{\wireLoad} = O(\log{n})$, and  
    we can conclude that the total round complexity is bounded from above by  
    \begin{align*}
        &O(d+q\log{n}) \cdot O(\log{n} \cdot \log{\wireLoad}) \cdot O(\maxWaitRet+\maxWaitStore) \\
        &=\lceil\maxWidth/n^2+\Delta/n \rceil2^{O(\sqrt{\log{n}}\log\log{n})}\text{.}\qedhere
    \end{align*} 
\end{proof}

\section{Deterministic LDC Decoding with Known Erasures}
\label{sec:derandom}
In this section we show a \emph{deterministic} LDC code resilient to $\alpha$-fraction of erasures, with the good properties required for our robust circuit computation in \Cref{alg:main}. 
Specifically, we provide an LDC with a deterministic decoding algorithm, that given the erasure locations in the codeword, succeeds in decoding any index of the message by performing $q$ queries to the codeword. We remark that knowledge of the erasure locations is crucial for deterministic local decoding, as otherwise an erasure of all of the $q \ll n$ queried indices leaves the decoder without any information and causes any local decoding procedure to fail. 
We further show that we can perform multiple such decodings (possibly of different codewords with different sets of erasures) while maintaining the ``smoothness'' of the queried locations, i.e., without querying any single index too many times across the different codewords. This property translates to maintaining the congestion induced by the LDC decoding of \Cref{alg:main}.

We base our construction on a standard (randomized)  Reed--Muller LDC (see, e.g., \cite{Yekhanin12}) and then show how to de-randomize its decoding algorithm
in a way that achieves good properties (i.e., the ones in \Cref{def:good_seed}) when employed in the \BulkRetrieve procedure.
Our analysis 
is based on a simple probabilistic method argument, which uses the properties of the Reed--Muller code and the fact that the corruption model is restricted to erasures.

\smallskip

We begin with a randomized LDC construction designed for the case of erasures. 
The following is a relatively standard analysis of Reed--Muller-based LDCs, with the required adaptation to erasures. See Proposition 2.3 in~\cite{Yekhanin12} for an analogous analysis assuming substitution corruption. 

\begin{lemma}[Erasure LDC]
\label{lem:reed_muller}
    Let $\delta \in (0,1)$ be a constant and let $r,d \geq 1$ be integers. Let $q$ be a power of a prime  such that $d \leq (1-\delta)(q-1) - 1$. Let $K = \binom{r+d}{d}\mathstrut$ and $N = q^r$. Then there is an LDC with an encoding function $\LDCEnc: \F_q^K \rightarrow \F_q^{N\mathstrut}$ and a randomized non-adaptive local decoding algorithm $\LDCDec:\F_q^N \times [K]\to \F_q$ that performs $q-1$ queries to the (possibly corrupted by erasures) codeword and satisfies the following properties:
    
    \begin{enumerate}
        \item \label{item:rm_threshold} If $\LDCDec$ queries at most $\delta (q-1)$ erased indices, then the decoding algorithm is guaranteed to succeed, namely, $\LDCDec(\LDCEnc(x),i)=x[i]$.
        \item \label{item:rm_smoothness} The probability of an index to be queried is $O(1/q)$ over the choice of the randomness string, i.e.,\@ the code is $(q/2N)$-smooth.
        \item \label{item:rm_expectation} For any constant $\gamma \in (0,1)$, if the (possibly corrupted) codeword has at most a $\gamma$-fraction of erasures, then  $\LDCDec$ queries at most $(1 + \frac{1}{N-1}) \gamma (q-1)$ erased indices, in expectation.
        
    \end{enumerate}
\end{lemma}
\begin{proof}
    The properties above are obtained by using a standard Reed-Muller code, and we give here a succinct proof for the sake of completeness. Fix an arbitrary set $S \subseteq \F_q^r$ of size $|S| = K$, and denote its elements as $S = \{s_1,\dots,s_K\}$.
    We define $\LDCEnc$ as follows: for a message $w \in \F_q^K$ to encode, there is a unique polynomial $f_w:\F_q^r \rightarrow \F_q$ on $r$ variables and degree at most~$d$ such that $f_w(s_i) = w[i]$ 
    (this follows by a simple interpolation argument). 
    We define the encoding of $w$ to be $\LDCEnc(w) = (f_w(x) \mid x \in \F_q^r)$, i.e., the evaluation of $f_w$ on all values of~$\F_q^r$. 

    The randomized function $\LDCDec$ receives as input an  index $i \in [K]$ to decode. We may assume that it also receives as input a randomness string $R$ which is used for any random choice of the function. 
    The function is defined as follows: First, we pick a uniformly random point $x \in \F_q^r\setminus \vec 0$. Let $L$ be the line  
    \[
    L = \{s_i + ax \mid a \in \F_q\}\text{,}
    \]
    and let $g_w: \F_q \rightarrow \F_q$ be the univariate polynomial  defined as
    \[
    g_w(a) = f_w(s_i+ax)\text{.}
    \]
    In particular, we note that the degree of $f_w$ bounds from above the degree of $g_w$. Since $f_w$ has degree at most $d$, it follows that $g_w$ also has degree at most $d$.
    We query from the codeword all points in $L \setminus \{s_i\}$. Since $|L| = q$, then indeed $\LDCDec$ queries $q-1$ coordinates of the codeword. Moreover, we note that the queried symbols  correspond exactly to the evaluation of $g_w$ on all points in~$\F_q \setminus \{0\}$.
    
    If $L \setminus \{s_i\}$ contains more than $\delta(q-1)$ erased indices, we return $\bot$. Otherwise, we have 
    at least $(1-\delta) (q-1)$ positions that  were not erased and their value is correct. 
    By a standard  interpolation argument, since $d\le(1-\delta)(q-1)-1$, there exists at most a single univariate polynomial~$h$ of degree at most~$d$ that agrees with all the non-erased points of~$g_w$.
    Moreover, since $g_w$ has degree at most~$d$, we have that this polynomial must be~$h = g_w$. Hence, the algorithm can compute $g_w:\F_q \rightarrow \F_q$, and return $g_w(0)$. Next, we prove the properties of this algorithm, described above.

    \begin{enumerate}
        \item This follows trivially by definition of $g_w$ that $g_w(0) = f_w(s_i)$.
        \item Since there is exactly one line going through any given two points in $\F_q^r$, then any two lines going through $s_i$ do not intersect on any other point. Since we choose a random line going through $s_i$, each point other than $s_i$ has has probability at most $(q-1)/(N-1) \geq q/2N$ 
        to be queried. Moreover, since $s_i$ is queried with probability zero, it follows that any given point has probability at most $q/2N$ to be queried. In other words, the code is $(q/2N)$-smooth.
        
        \item Consider the lines $L_1,\dots,L_{\frac{N-1}{q-1}}$ going through $s_i$. As stated in the previous paragraph, any two such lines only intersect at $s_i$. Since the total number of erasures is $\gamma N$, it follows that for a random line $L \in \{L_1,\dots,L_{\frac{N-1}{q-1}}\}$, the number of erasures in $L \setminus \{s_i\}$ is $\gamma N \cdot \frac{q-1}{N-1} = (1 + \frac{1}{N-1}) \gamma (q-1)$, in expectation.
        \qedhere
    \end{enumerate}
\end{proof}

We next show a specific choice of parameters for the code presented in \Cref{lem:reed_muller}, that fits best to our \Cref{alg:main}; this is the code used as the LDC instantiation described in \Cref{sec:primitives}.
\begin{lemma}
\label{lem:LDC_muller_augmented}
    Let $\delta \in (0,1)$ be a constant and let $N,q$ be sufficiently large integers, where $2^{\sqrt{\log{N}}} \leq q \leq 2^{2\sqrt{\log{N}}}$ is a power of a prime, and $N$ is a power of $q$. \strut There exists an LDC with encoding function $\LDCEnc: \F_q^K \rightarrow \F_q^N$ with $K = N \cdot 2^{-O(\sqrt{\log{N}}\log\log{N})}$, whose randomized non-adaptive local decoding algorithm satisfies Properties~(\ref{item:rm_threshold})--(\ref{item:rm_expectation}) of \Cref{lem:reed_muller}. 
\end{lemma}
\begin{proof}
    We define our LDC by instantiating the code in \Cref{lem:reed_muller} with parameters $\delta',r',d',q'$ set as follows:  $q' = q$, $\delta' = \delta$, 
    $d' = \lfloor (1-\delta')(q'-1) \rfloor - 1$, and $r' = \log_{q}{N}$. In particular, we have that 
    \begin{align*}
        K &= \binom{r'+d'}{d'} =\binom{r'+d'}{r'} \geq
        \left(\frac{d'}{r'}\right)^{r'} \\
        &\geq \left(\frac{(1-\delta')q}{2r'}\right)^{r'} = \frac{N}{((2/(1-\delta))\log_q(N))^{\log_q{N}}} \\ 
        &= \frac{N}{2^{O(\sqrt{\log{N}}\log\log{N})}}\text{.}
    \end{align*}

    Where the final equality follows because $q = 2^{\Theta(\sqrt{\log{N}})}$, hence $\log_q{N} = \Theta(\sqrt{\log{N}})$.
    The local decoding algorithm $\LDCDec$ is defined the same as in \Cref{lem:reed_muller}. Hence Properties~(\ref{item:rm_threshold})--(\ref{item:rm_expectation}), proven in \Cref{lem:reed_muller}, also hold for this code.
\end{proof}

Next, we show a \emph{deterministic}
``bulk decoding'' algorithm 
for the code we constructed in \Cref{lem:LDC_muller_augmented}. 
Specifically, in 
bulk-decoding, we want to locally decode a set of $P$ codewords instead of a single codeword. 

The possibility of derandomizing the LDC decoding in our setting relies on two important differences from standard LDC decoding. First, we only assume erasures (corresponding to node crashes). Second, the indices that are erased are known to the decoder in advance (this fits the information a node has in \Cref{alg:main} about crashed nodes). 
This knowledge allows deterministic decoding: instead of randomly querying $q$ indices (nodes) and trying to decode according to the potentially erased symbols they hold, the decoder goes over the randomness strings given to the decoding algorithm and checks \emph{in advance} whether that randomness results in querying sufficiently many non-crashed nodes so that the decoding succeeds according to Property~(\ref{item:rm_threshold}) of \Cref{lem:reed_muller}.
If not, it can skip to the next randomness string until a good randomness string is found.

This approach can be highly inefficient in our setting when considering bulk-decoding of $P$ different codewords. Suppose that a node wishes to decode $P$ indices of the same codeword, then if we derandomize each decoding separately, all the decodings will query the same set of $q$ nodes because the same set would result from the above exhaustive search for a good randomness string, scaling-up the congestion by~$P$. 

In the following we show how to derandomize bulk-decoding by scaling the congestion by $(P\log N)/N$, i.e., keeping the smoothness of the decoding scheme, up to $(\log N)/N$ terms. 
It is important to mention that each individual decoding is guaranteed to succeed, that is, we only choose a randomness string if it satisfies Property~(\ref{item:rm_threshold}) of \Cref{lem:reed_muller}.
In particular, we show that such randomness strings with good properties \emph{exist}, which means that all nodes can locally compute these strings and use them in~\Cref{alg:main} without the need for randomness. 

Let us set up the notation for the derandomization proof. 
The decoding algorithm has access to $P$ possibly corrupted codewords, $C_1,\ldots, C_p$, where $C_j=\LDCEnc(m_j)$ of some message $m_j$, up to erasures.
For each such codeword~$C_j$ with $j\in[P]$, the algorithm is given as input an index $i_j\in[K]$ to decode.
Further, the algorithm is given a set $I_j\subseteq [N]$ of erased codeword symbols of~$C_j$, where $|I_j|<\alpha N$.

The algorithm returns a collection of randomness strings $R_1,\ldots, R_P$ that imply a respective set of queries $Q_1,\dots,Q_P$  (where for any $j\in[P]$,  we have $Q_j\subset[N]$ and $|Q_j|=q$), with the following properties:
(a) no index $\ell\in [N]$ appears in more than $ O(\lceil Pq/N \rceil \log{N})$ sets of $Q_1,\dots,Q_P$, and 
(b) for any $j\in[P]$, we can locally decode the symbol in index $i_j$ of~$m_j$ by querying the $q$ symbols of~$C_j$ specified by~$Q_j$.

\begin{lemma}
\label{lem:compute_randomness_deter}
There is a deterministic algorithm $\LDCDetDec$, which for any integer $P\ge1$, any constant value $0 < \delta' < \delta$, any sets $I_1,\dots,I_P \subseteq [N]$ such that $|I_j| \leq \delta' N$ for all $1\leq j\leq P$, and  any set of indices  $i_1,\dots,i_P \in [K]$, returns randomness strings $R_1,\dots,R_P$ with the following properties: Let $Q_j \subseteq [N]$ be the queries made by $\LDCDec(.,i_j)$ using randomness strings $R_j$, where $\LDCDec$ is the decoding function of \Cref{lem:LDC_muller_augmented}, then:
\begin{enumerate}
     \item \label{item:query_const_cong} For every $\ell \in [N]$, it holds that $|\{j \in [P] \mid \ell \in Q_j\}| = O(\lceil Pq/N \rceil \log{N})$.
     \item \label{item:query_const_correctness} 
     Fix $j\in[P]$ and $m_j\in[q]^k$. Let $\LDCEnc$ be the encoding function from \Cref{lem:LDC_muller_augmented}, and let $C_j$ be $\LDCEnc(m_j)$ with indices in $I_j$ erased. Then when applying $\LDCDec(C_j,i_j)$ with randomness string $R_j$, the decoding returns the symbol in index $i_j$ of~$m_j$. 
 \end{enumerate}
\end{lemma}
\begin{proof}

    We show the existence of $R_1,...,R_P$ satisfying the properties of the lemma using the probabilistic method. This implies a deterministic algorithm in which we can enumerate all possible randomness strings $R_1,\dots,R_P$ and choose the lexicographically smallest set of strings $(R_1,\dots,R_P)$ satisfying the properties above.

    We choose each $R_j$ using the following random process. 
    Let $M = c\log{K}$ for a sufficiently large constant $c >0$. 
    For any $j\in[P]$, let
   $R_j^1,\dots,R_j^M$
    be uniformly random strings in $\{0,1\}^{\poly(K)}$. 
    Let $Q_j^t\subset[N]$ be the queries that $\LDCDec$ performs given the decoding index $i_j$ and randomness string $R_j^t$. Recall that $\LDCDec$ is non-adaptive which means that these indices do not depend on the codeword $C_j$ or the erased symbols~$I_j$.
    Set $R_j=R_j^t$ for the first $t \in [M]$ such that $|Q_j^t \cap I_j| \leq \delta (q-1)$, or $\bot$ if no such string exists. If $R_j \neq \bot$, then item~(\ref{item:query_const_correctness}) of the lemma is guaranteed by Property~(\ref{item:rm_threshold}) of \Cref{lem:LDC_muller_augmented}. Next, we show that $R_j \neq \bot$ with high probability.

    By Property~(\ref{item:rm_expectation}) of \Cref{lem:LDC_muller_augmented}, we have that $E[|Q_j^t \cap I_j|] \leq (1+\frac{1}{N-1})\delta' (q-1)$. By Markov's inequality, we have that 
    \[
    \Pr\left(|Q_j^t \cap I_j| 
    \geq \delta(q-1)\right) 
    \leq  
    \frac{(1+\frac{1}{N-1})\delta' (q-1)}{\delta(q-1)}
    \]
    Since $(1+\frac{1}{N-1})\delta'/\delta$
    is bounded by a constant smaller than $1$ for sufficiently large $N$ and since we have $M = O(\log{K})$ independent such strings by a correct choice of constants, we can guarantee that at least one $Q_j^t$ satisfies $|Q_j^t \cap I_j| \leq \delta'(q-1)$, with high probability.

    Finally, we show that Property~(\ref{item:query_const_cong}) holds for our choice of $R_1,\dots,R_P$. We prove the slightly stronger claim: that for every $\ell \in [N]$, it holds that 
    \[
    |\{(j,t) \in [P] \times [M] \mid \ell \in Q^t_j\}| = O(\lceil Pq/N \rceil \log{N})\text{.}
    \]
    
    Informally, we can think of the queries as a sort of bins-into-balls process. Each local decoding call with randomness string $R_j^t$ throws $q$ balls into $N$ bins. However, for a single call, the locations of the $q$ balls are \emph{dependent}.
        Nevertheless, by the smoothness  property of the LDC (Property~(\ref{item:rm_smoothness}) of \Cref{lem:LDC_muller_augmented}), any specific bin $i$ gets a ball with probability at most $q/2N$, 
        and this probability is independent across different $R_j$'s. 
         
        Fix some index~$\ell \in [N]$. For $j\in[P]$ and $t\in[M]$, denote by $Y_j^t$ the indicator that equals $1$ if and only if  $\ell \in Q_j^t$. 
        Set $Y=\sum_{t=1}^{M} \sum_{j=1}^{P}Y_j^t$. 
        By Property~(\ref{item:rm_smoothness}) of \Cref{lem:LDC_muller_augmented}, we have
     \begin{align*}
        &\forall (j,t), \quad\Pr(Y_j^t=1)\le q/2N, 
        \intertext{hence, by linearity of expectation,}
        &E[Y] \le  PM \cdot q/2N\text{.}
     \end{align*}
     Then, we use a Chernoff inequality  (\Cref{thm:chernoff}(3)) 
     on the $P$ independent indicators, bound the probability that  index~$\ell$ is queried too many times. 
     Specifically,  
     \[
     \Pr( Y > 6\lceil P q/(2N)\rceil M) =  
\Pr( Y > 6c\lceil P q/(2N)\rceil \log{N}) \le 2^{-6c \log N}\text{,}
\]
which follows since $\lceil P q/2 N\rceil\ge 1$.
     Taking the union bound over all indices $\ell \in [N]$ and setting $c$ to be sufficiently large, we get that all indices $i \in [N]$ are queried less than $6c \lceil P q/2 N\rceil \log{N}$ times, with high probability. 
\end{proof}

Our final task is to complete the proof \Cref{lem:good_seeds_computation} (restated below). Recall that in the LDC instantiation of our circuit computation algorithm, we set $N = n$, i.e., the LDC codeword length is set to be the size of the network. 
For convenience, let us also recall the definition of a collection of good randomness strings (\Cref{def:good_seed}).

\GoodSeedDef*
\GoodQueryCompute*
\begin{proof}
Recall that in the \Store procedure, each node $v_i$ stores the $i$-th symbol of each stored LDC codeword. Let $I$ be the set of indices stored in the crashed nodes at the start of $\BulkRetrieve$, i.e., $i \in I$ if and only if $v_i \in V \setminus \alive$. Conceptually, $I$ corresponds to our erasure set when retrieving a stored value.

We apply \Cref{lem:compute_randomness_deter} with the following parameters: we set $P = 2^\ell|W_j|$, and set $\delta' = (\alpha+\delta)/2$, where $\alpha$ is the fraction of node crashes the adversary can cause and $\delta$ is the distance parameter of the LDC. Recall that we set $\delta$ to be a constant $\alpha < \delta$, hence $\alpha < \delta' < \delta$.  For each index $w \in W_j$, we set $2^\ell$ decoding indices $i_{j,w,1},\dots,i_{j,w,2^{\ell}} = w$ and we set all erasure sets $I_1,\dots,I_P$ to be equal to $I$.

Algorithm $\LDCDetDec$ returns a collection of randomness strings $\{R_{v_j,w,k}\}_{w \in W_j,k \in 2^\ell}$ such that the resulting query sets $\{Q_{v_j,w,k}\}_{w \in W_j,k \in 2^\ell}$ have properties as described in \Cref{lem:compute_randomness_deter}. We conclude by showing that these randomness strings are good according to \Cref{def:good_seed}.  

\textbf{Satisfying \Cref{def:good_seed}(\ref{item:good_seed_cong}):}  As stated at the start of the proof, in an invocation of the $\Store$ procedure, each node $v_i$ receives only the $i$'th symbol of each codeword. By  \Cref{lem:compute_randomness_deter}(\ref{item:query_const_cong}), each index $i \in [n]$ is contained in at most $O(\lceil 2^\ell |W_j|q/n \rceil \log{n})$ query sets from $Q_{v_j,w,i}$, it follows that each node is queried at most $O(\lceil 2^\ell |W_j|q/n \rceil \log{n})$ times by $v_j$.

\textbf{Satisfying \Cref{def:good_seed}(\ref{item:good_seed_correctness}):} \Cref{lem:compute_randomness_deter}(\ref{item:query_const_correctness}) promises that for any $w \in W_j$ and $i \in [2^\ell]$, given all queries of $Q_{v_j,w,i}$ to the codeword $U_{(w)}$, one can decode the symbol $w$ from $U_{(w)}$. In particular, if node $v_j$ receives responses to all queries made to nodes of $\alive$, it can decode $w$ from the respective stored codeword $U_{(w)}$ and thus the \Retrieve invocation succeeds.
\end{proof}

 \section*{Acknowledgments} 
 We would like to thank Merav Parter and Noga Ron-Zewi for helpful discussions. R. Gelles is supported in part by the United States -- Israel Binational Science Foundation (BSF), grant No.\@ 2020277. O. Fischer is supported in part by the Israel Science Foundation, grant No. 1042/22 and 800/22. K. Censor-Hillel is supported in part by the Israel Science Foundation, grant No. 529/23.

\bibliography{bibliography}

\appendix

\section{Chernoff Bounds}

\begin{theorem}[Chernoff inequality for independent Bernoulli variables] \label{thm:chernoff}
Let $X_1,\ldots, X_n$ be mutually independent 0--1 random variables with $\Pr(X_i=1)=p_i$. Let $X=\sum_{i=1}^n X_i$ and set $\mu=E[X]$. The following holds,
\begin{enumerate}
    \item for any $\delta>0$, 
    \( 
    \Pr(X\ge (1+\delta)\mu) \le 
        \left( 
            \frac{e^\delta}{(1+\delta)^{(1+\delta)}}
        \right)^{\mu} 
    \)
    \item for $0<\delta\le 1$,
    \(
    \Pr(X\ge (1+\delta)\mu) \le
        e^{-\mu\delta^2/3}
    \)
    \item for $R \ge 6\mu$,
    \(
    \Pr(X\ge R) \le 2^{-R}
    \)
\end{enumerate}
\end{theorem}
For proof, see Theorem~4.4 in \cite{MU17book}.
\end{document}

%% file: macros.tex
\theoremstyle{plain}
\newtheorem{theorem}{Theorem}[section]

\newtheorem{lemma}[theorem]{Lemma}
\newtheorem{observation}[theorem]{Observation}

\newtheorem{corollary}[theorem]{Corollary}

\newtheorem{definition}[theorem]{Definition}

\algblock{ParFor}{EndParFor}
\algnewcommand\algorithmicparfor{\textbf{parallel for}}
\algnewcommand\algorithmicpardo{\textbf{do}}
\algnewcommand\algorithmicendparfor{\textbf{end\ parallel for}}
\algrenewtext{ParFor}[1]{\algorithmicparfor\ #1\ \algorithmicpardo}
\algrenewtext{EndParFor}{\algorithmicendparfor}

\algnewcommand\algorithmicforeach{\textbf{for each}}
\algdef{S}[FOR]{ForEach}[1]{\algorithmicforeach\ #1\ \algorithmicdo}

\makeatletter
\newcommand{\multiline}[1]{%
  \begin{tabularx}{\dimexpr\linewidth-\ALG@thistlm}[t]{@{}X@{}}
    #1
  \end{tabularx}
}
\makeatother

\DeclareMathOperator{\poly}{poly}

\newcommand{\F}{\mathbb{F}}
\newcommand{\Fp}{\mathbb{F}_p}

\newcommand{\last}{\mathsf{last}}

\DeclareMathOperator{\Hamm}{\mathrm{Hamm}}

\newcommand{\congest}{\textsf{Congest}\xspace}

\newcommand{\cliquefull}{\textsf{Congested Clique}\xspace}

\newcommand{\Allocate}{\textsc{Allocate}\xspace}
\newcommand{\Store}{\textsc{Store}\xspace}
\newcommand{\Retrieve}{\textsc{Retrieve}\xspace}
\newcommand{\BulkRetrieve}{\textsc{BulkRetrieve}\xspace}

\newcommand{\alive}{\mathcal{A}}

\newcommand{\LDCEnc}{\mathsf{Enc}\xspace}
\newcommand{\LDCDec}{\mathsf{Dec}\xspace}

\newcommand{\LDCDetDec}{\mathsf{DetDec}}

\newcommand{\inStep}{\textsc{\footnotesize AttemptDoubling}}
\newcommand{\outStep}{\textsc{\footnotesize NodeDoubling}}

\newcommand{\layerStep}{\textsc{\footnotesize{layer}}}
\newcommand{\layer}{\mathsf{layer}}

    \mathchardef\mhyphen="2D
    \newcommand{\fanin}{\mathsf{fan\mhyphen in}}
    \newcommand{\fanout}{\mathsf{fan\mhyphen out}}
    \newcommand{\fantotal}{\mathsf{fan}}
    \newcommand{\wireLoad}{\Lambda}

    \newcommand{\gates}{\mathsf{gates}}
    \newcommand{\wires}{\mathsf{wires}}

\newcommand{\maxWidth}{\omega}

\newcommand{\maxWaitRet}{\mathsf{maxRetTime}}
\newcommand{\maxWaitStore}{\mathsf{maxStoreTime}}



%% file: circuit.tex
\begin{tikzpicture}[yscale=0.75,
  node distance=3ex and 5ex,
  gate/.style={draw, rectangle, minimum size=5ex},
  input/.style={draw, circle, minimum size=2ex},
  output/.style={draw,  circle, minimum size=2ex},
  wire/.style={->, thick},
  wire2/.style={dashed,->, very thick}
]

\node[input] (in1) at (0,0) {in$_1$};
\node[input] (in2) at (1,0) {in$_2$};
\node[input] (in3) at (2,0) {in$_3$};
\node[input] (in4) at (3,0) {in$_4$};

\node[gate] (g1) at (0,2.5) {$g_1$};
\node[gate] (g2) at (2.5,2.5) {$g_2$};

\node[gate] (g3) at (2,5) {$g_3$};

\node[output] (out2) at (2,7) {out$_2$};
\node[output] (out1) at (0,5)  {out$_1$};

\draw[wire] (in1) -- (g1);
\draw[wire] (in2) -- (g2);
\draw[wire2] (g1) -- (g3);
\draw[wire2] (g2) -- (g3);
\draw[wire] (g3) -- (out2);
\draw[wire] (in1) -- (g2);
\draw[wire] (in3) -- (g2);
\draw[wire] (in4) -- (g2);
\draw[wire] (in2) -- (g3);
\draw[wire2] (g1) -- (out1);

\node at (-1.75, 0) {Layer 0};
\node at (-1.75, 2.5) {Layer 1};
\node at (-1.75, 5) {Layer 2};
\node at (-1.75, 7) {Layer 3};

\draw[Latex-,thick] (3,1) -- (4,1);
\node at (5,1) {$\wires(0)$};

\draw[dashed,Latex-,very thick] (3,3.85) -- (4,3.85);
\node at (5,3.85) {$\wires(1)$};

\end{tikzpicture}

%% file: FaultTolerentComputing.bbl
\begin{thebibliography}{10}

\bibitem{AshkenaziGL22}
Yagel Ashkenazi, Ran Gelles, and Amir Leshem.
\newblock Noisy beeping networks.
\newblock {\em Information and Computation}, 289:104925, 2022.

\bibitem{AttiyaWelch}
Hagit Attiya and Jennifer~L. Welch.
\newblock {\em Distributed computing - fundamentals, simulations, and advanced
  topics {(2.} ed.)}.
\newblock Wiley series on parallel and distributed computing. Wiley, 2004.

\bibitem{AMPV22}
John Augustine, Anisur~Rahaman Molla, Gopal Pandurangan, and Yadu Vasudev.
\newblock Byzantine connectivity testing in the congested clique.
\newblock In {\em 36th International Symposium on Distributed Computing
  (DISC)}, volume 246, pages 7:1--7:21, 2022.

\bibitem{BKM20}
Philipp Bamberger, Fabian Kuhn, and Yannic Maus.
\newblock Efficient deterministic distributed coloring with small bandwidth.
\newblock In {\em {ACM} Symposium on Principles of Distributed Computing
  (PODC)}, pages 243--252, 2020.

\bibitem{BKO22}
Aviv Bick, Gillat Kol, and Rotem Oshman.
\newblock Distributed zero-knowledge proofs over networks.
\newblock In {\em Proceedings of the 2022 {ACM-SIAM} Symposium on Discrete
  Algorithms (SODA)}, pages 2426--2458, 2022.

\bibitem{CDKL21}
Keren Censor{-}Hillel, Michal Dory, Janne~H. Korhonen, and Dean Leitersdorf.
\newblock Fast approximate shortest paths in the congested clique.
\newblock {\em Distributed Computing}, 34(6):463--487, 2021.

\bibitem{Censor-HillelFG22}
Keren Censor{-}Hillel, Orr Fischer, Fran{\c{c}}ois~Le Gall, Dean Leitersdorf,
  and Rotem Oshman.
\newblock Quantum distributed algorithms for detection of cliques.
\newblock In {\em 13th Innovations in Theoretical Computer Science Conference
  (ITCS)}, volume 215, pages 35:1--35:25, 2022.

\bibitem{Censor-HillelFG20}
Keren Censor{-}Hillel, Orr Fischer, Tzlil Gonen, Fran{\c{c}}ois~Le Gall, Dean
  Leitersdorf, and Rotem Oshman.
\newblock Fast distributed algorithms for girth, cycles and small subgraphs.
\newblock In {\em 34th International Symposium on Distributed Computing
  (DISC)}, volume 179, pages 33:1--33:17, 2020.

\bibitem{Censor-HillelGL20}
Keren Censor{-}Hillel, Fran{\c{c}}ois~Le Gall, and Dean Leitersdorf.
\newblock On distributed listing of cliques.
\newblock In {\em Symposium on Principles of Distributed Computing (PODC)},
  pages 474--482, 2020.

\bibitem{Censor-HillelH24}
Keren Censor{-}Hillel and Einav Huberman.
\newblock Near-optimal resilient labeling schemes.
\newblock In {\em 28th International Conference on Principles of Distributed
  Systems (OPODIS)}, volume 324, pages 35:1--35:22, 2024.

\bibitem{CKPLPS15}
Keren Censor-Hillel, Petteri Kaski, Janne~H. Korhonen, Christoph Lenzen, Ami
  Paz, and Jukka Suomela.
\newblock Algebraic methods in the congested clique.
\newblock In {\em Proceedings of the 2015 ACM Symposium on Principles of
  Distributed Computing (PODC)}, page 143–152, 2015.

\bibitem{CS25}
Keren Censor-Hillel and Pedro Soto.
\newblock Computing in a faulty congested clique.
\newblock {\em CoRR}, abs/2505.11430, 2025.

\bibitem{CFGUZ19}
Yi{-}Jun Chang, Manuela Fischer, Mohsen Ghaffari, Jara Uitto, and Yufan Zheng.
\newblock The complexity of ({\(\Delta\)}+1) coloring in congested clique,
  massively parallel computation, and centralized local computation.
\newblock In {\em Proceedings of the 2019 {ACM} Symposium on Principles of
  Distributed Computing (PODC)}, pages 471--480, 2019.

\bibitem{MR23}
David Cifuentes{-}N{\'{u}}{\~{n}}ez, Pedro Montealegre, and Ivan Rapaport.
\newblock Recognizing hereditary properties in the presence of byzantine nodes.
\newblock {\em CoRR}, abs/2312.07747, 2023.

\bibitem{CoyCDM23}
Sam Coy, Artur Czumaj, Peter Davies, and Gopinath Mishra.
\newblock Optimal (degree+1)-coloring in congested clique.
\newblock In {\em 50th International Colloquium on Automata, Languages, and
  Programming (ICALP)}, volume 261, pages 46:1--46:20, 2023.

\bibitem{CzumajDP21}
Artur Czumaj, Peter Davies, and Merav Parter.
\newblock Simple, deterministic, constant-round coloring in congested clique
  and {MPC}.
\newblock {\em {SIAM} J. on Computing}, 50(5):1603--1626, 2021.

\bibitem{Davies23a}
Peter Davies.
\newblock Optimal message-passing with noisy beeps.
\newblock In {\em Proceedings of the 2023 {ACM} Symposium on Principles of
  Distributed Computing (PODC)}, pages 300--309, 2023.

\bibitem{DLP12}
Danny Dolev, Christoph Lenzen, and Shir Peled.
\newblock ``tri, tri again'': Finding triangles and small subgraphs in a
  distributed setting.
\newblock In {\em Distributed Computing}, 2012.

\bibitem{DKO13}
Andrew Drucker, Fabian Kuhn, and Rotem Oshman.
\newblock On the power of the congested clique model.
\newblock In {\em {ACM} Symposium on Principles of Distributed Computing
  (PODC)}, pages 367--376, 2014.

\bibitem{FischerGKO18}
Orr Fischer, Tzlil Gonen, Fabian Kuhn, and Rotem Oshman.
\newblock Possibilities and impossibilities for distributed subgraph detection.
\newblock In {\em Proceedings of the 30th on Symposium on Parallelism in
  Algorithms and Architectures (SPAA)}, pages 153--162, 2018.

\bibitem{fischer2022explicit}
Orr Fischer, Rotem Oshman, and Dana Shamir.
\newblock Explicit space-time tradeoffs for proof labeling schemes in graphs
  with small separators.
\newblock In {\em 25th International Conference on Principles of Distributed
  Systems (OPODIS 2021)}, 2022.

\bibitem{FP23}
Orr Fischer and Merav Parter.
\newblock Distributed {CONGEST} algorithms against mobile adversaries.
\newblock In {\em Proceedings of the 2023 {ACM} Symposium on Principles of
  Distributed Computing (PODC)}, pages 262--273, 2023.

\bibitem{FP25}
Orr Fisher and Merav Parter.
\newblock All-to-all communication with mobile edge adversary: Almost linearly
  more faults, for free.
\newblock {\em CoRR}, abs/2505.05735, 2025.

\bibitem{GKKM25arXiv}
Pawel Garncarek, Dariusz~R. Kowalski, Shay Kutten, and Miguel~A. Mosteiro.
\newblock Beeping deterministic congest algorithms in graphs, 2025.

\bibitem{GhaffariP16}
Mohsen Ghaffari and Merav Parter.
\newblock {MST} in log-star rounds of congested clique.
\newblock In {\em Proceedings of the 2016 {ACM} Symposium on Principles of
  Distributed Computing (PODC)}, pages 19--28, 2016.

\bibitem{GPP25}
Alex~B. Grilo, Ami Paz, and Mor Perry.
\newblock Distributed non-interactive zero-knowledge proofs.
\newblock {\em CoRR}, abs/2502.07594, 2025.

\bibitem{HegemanPPSS15}
James~W. Hegeman, Gopal Pandurangan, Sriram~V. Pemmaraju, Vivek~B. Sardeshmukh,
  and Michele Scquizzato.
\newblock Toward optimal bounds in the congested clique: Graph connectivity and
  {MST}.
\newblock In {\em Proceedings of the 2015 {ACM} Symposium on Principles of
  Distributed Computing (PODC)}, pages 91--100, 2015.

\bibitem{IzumiG17}
Taisuke Izumi and Fran{\c{c}}ois~Le Gall.
\newblock Triangle finding and listing in {CONGEST} networks.
\newblock In {\em Proceedings of the {ACM} Symposium on Principles of
  Distributed Computing (PODC)}, pages 381--389, 2017.

\bibitem{Jurdzinski018}
Tomasz Jurdzinski and Krzysztof Nowicki.
\newblock {MST} in \emph{O}(1) rounds of congested clique.
\newblock In {\em Proceedings of the Twenty-Ninth Annual {ACM-SIAM} Symposium
  on Discrete Algorithms (SODA)}, pages 2620--2632, 2018.

\bibitem{KT00}
Jonathan Katz and Luca Trevisan.
\newblock On the efficiency of local decoding procedures for error-correcting
  codes.
\newblock In {\em Proceedings of the Thirty-Second Annual ACM Symposium on
  Theory of Computing (STOC)}, page 80–86, 2000.

\bibitem{Korhonen16}
Janne~H. Korhonen.
\newblock Deterministic {MST} sparsification in the congested clique.
\newblock {\em CoRR}, abs/1605.02022, 2016.

\bibitem{KM23}
Manish Kumar.
\newblock Fault-tolerant graph realizations in the congested clique, revisited.
\newblock In {\em Distributed Computing and Intelligent Technology}, pages
  84--97, 2023.

\bibitem{KMS22}
Manish Kumar, Anisur~Rahaman Molla, and Sumathi Sivasubramaniam.
\newblock Fault-tolerant graph realizations in the congested clique.
\newblock In {\em Algorithmics of Wireless Networks}, pages 108--122, Cham,
  2022.

\bibitem{Lenzen13}
Christoph Lenzen.
\newblock Optimal deterministic routing and sorting on the congested clique.
\newblock In {\em Proceedings of the 2013 ACM Symposium on Principles of
  Distributed Computing (PODC)}, page 42–50, 2013.

\bibitem{LPPP05}
Zvi Lotker, Boaz Patt{-}Shamir, Elan Pavlov, and David Peleg.
\newblock Minimum-weight spanning tree construction in \emph{O}(log log
  \emph{n}) communication rounds.
\newblock {\em {SIAM} J. on Computing}, 35(1):120--131, 2005.

\bibitem{Lynch96}
Nancy~A. Lynch.
\newblock {\em Distributed Algorithms}.
\newblock Morgan Kaufmann, 1996.

\bibitem{MU17book}
Michael Mitzenmacher and Eli Upfal.
\newblock {\em Probability and computing: Randomization and probabilistic
  techniques in algorithms and data analysis}.
\newblock Cambridge university press, 2017.

\bibitem{NPY20}
Moni Naor, Merav Parter, and Eylon Yogev.
\newblock The power of distributed verifiers in interactive proofs.
\newblock In {\em Proceedings of the 2020 {ACM-SIAM} Symposium on Discrete
  Algorithms (SODA)}, pages 1096--115, 2020.

\bibitem{Nowicki21a}
Krzysztof Nowicki.
\newblock A deterministic algorithm for the {MST} problem in constant rounds of
  congested clique.
\newblock In {\em 53rd Annual {ACM} {SIGACT} Symposium on Theory of Computing
  (STOC)}, pages 1154--1165, 2021.

\bibitem{Pandurangan0S18}
Gopal Pandurangan, Peter Robinson, and Michele Scquizzato.
\newblock On the distributed complexity of large-scale graph computations.
\newblock In {\em Proceedings of the 30th on Symposium on Parallelism in
  Algorithms and Architectures (SPAA)}, pages 405--414, 2018.

\bibitem{Parter18}
Merav Parter.
\newblock (delta+1) coloring in the congested clique model.
\newblock In {\em 45th International Colloquium on Automata, Languages, and
  Programming (ICALP)}, volume 107, pages 160:1--160:14, 2018.

\bibitem{ParterS18}
Merav Parter and Hsin{-}Hao Su.
\newblock Randomized {(Delta+1)}-coloring in {O(log* Delta)} congested clique
  rounds.
\newblock In {\em 32nd International Symposium on Distributed Computing
  (DISC)}, volume 121, pages 39:1--39:18, 2018.

\bibitem{peleg2000distributed}
David Peleg.
\newblock {\em Distributed computing: a locality-sensitive approach}.
\newblock SIAM, 2000.

\bibitem{Spielman96}
Daniel~A. Spielman.
\newblock Highly fault-tolerant parallel computation.
\newblock In {\em Proceedings of 37th Conference on Foundations of Computer
  Science}, pages 154--163, 1996.

\bibitem{Yekhanin12}
Sergey Yekhanin.
\newblock Locally decodable codes.
\newblock {\em Foundations and Trends® in Theoretical Computer Science},
  6(3):139--255, 2012.

\end{thebibliography}
